\documentclass[10pt,journal,compsoc]{IEEEtran}
\ifCLASSOPTIONcompsoc
  \usepackage[nocompress]{cite}
\else
  \usepackage{cite}
\fi
\ifCLASSINFOpdf
  \usepackage[pdftex]{graphicx}
  \graphicspath{{figs/}}
  \DeclareGraphicsExtensions{.pdf,.esp,.jpg}
\else
  \usepackage[dvips]{graphicx}
  \graphicspath{{figs/}}
  \DeclareGraphicsExtensions{.pdf,.esp,.jpg}
\fi
\usepackage{amsmath}
\newtheorem{proof}{Proof}
\newtheorem{theorem}{Theorem}

\usepackage[ruled,vlined,linesnumbered]{algorithm2e}

\ifCLASSOPTIONcompsoc
 \usepackage[caption=false,font=footnotesize,labelfont=sf,textfont=sf]{subfig}
\else
 \usepackage[caption=false,font=footnotesize]{subfig}
\fi
\begin{document}
%
\title{Equi-depth Histogram Construction for Big Data with Quality Guarantees}
%
%
%
%

\author{Burak~Y{\i}ld{\i}z,
        Tolga~B\"{u}y\"{u}ktan{\i}r,
        and~Fatih~Emekci
\IEEEcompsocitemizethanks{\IEEEcompsocthanksitem The authors are with the Department
of Computer Engineering, Turgut Ozal University, Ankara,
Turkey. E-mail: \{yildizb, tbuyuktanir,\protect\\femekci\}@turgutozal.edu.tr}
\thanks{}}

%
%

\markboth{Equi-depth Histogram Construction for Big Data with Quality Guarantees}%
{}
%



\IEEEtitleabstractindextext{%
\begin{abstract}
The amount of data generated and stored in cloud systems has been increasing exponentially. The examples of data include user generated data, machine generated data as well as data crawled from the Internet. There have been several frameworks with proven efficiency to store and process the petabyte scale data such as Apache Hadoop, HDFS and several NoSQL frameworks. These systems have been widely used in industry and thus are subject to several research. The proposed data processing techniques should be compatible with the above frameworks in order to be practical. One of the key data operations is deriving equi-depth histograms as they are crucial in understanding the statistical properties of the underlying data with many applications including query optimization. In this paper, we focus on approximate equi-depth histogram construction for big data and propose a novel merge based histogram construction method with a histogram processing framework which constructs an equi-depth histogram for a given time interval. The proposed method constructs approximate equi-depth histograms by merging exact equi-depth histograms of partitioned data by guaranteeing a maximum error bound on the number of items in a bucket (bucket size) as well as any range on the histogram.

\end{abstract}

}

\maketitle

\IEEEdisplaynontitleabstractindextext

%
\IEEEpeerreviewmaketitle

\IEEEraisesectionheading{\section{Introduction}\label{sec:intro}}
\IEEEPARstart{T}{he} data generated and stored by enterprises are in the orders of terabytes or even petabytes~\cite{logothetis2010stateful,thusoo2010data,thusoo2010hive}. We can classify the source of the data in the following groups: machine generated data (a.k.a logs), social media data, transactional data and data generated by medical and wearable devices. Processing the produced data and deriving results are critical in decision making and thus the most important competitive power for the data owner. Therefore, handling such big datasets in an efficient way is a clear need for many institutions. Hadoop MapReduce~\cite{hadoop,dean2010mapreduce} is a big data processing framework that has rapidly become the standard method to deal with data bombarding in both industry and academia~\cite{dittrich2010hadoopplusplus,gates2009building,jindal2011trojan,zaharia2008improving,isard2007dryad,thusoo2010hive}. The main reasons of such strong adoption are the ease-of-use, scalability, failover and open-source properties of Hadoop framework.  After the wide distribution, many research works (from industry and academia) have focused on improving the performance of Hadoop MapReduce jobs in many aspects such as different data layouts~\cite{jindal2011trojan,floratou2011column,lin2011llama}, join algorithms~\cite{afrati2010optimizing,blanas2010comparison,okcan2011processing}, high-level query languages~\cite{gates2009building,isard2007dryad,thusoo2010hive}, failover algorithms~\cite{quiane2011rafting}, query optimization techniques~\cite{wu2011query,babu2010towards,herodotou2011profiling,jahani2011automatic}, and indexing techniques~\cite{dittrich2010hadoopplusplus,jiang2010performance,dittrich2012only}. 

In today's fast-paced business environment, obtaining results quickly represents a key desideratum for \textit{Big Data Analytics} \cite{jindal2011trojan}. For most applications on large datasets, performing careful sampling and computing early results from such samples provide a fast and effective way to obtain approximate results within the predefined level of accuracy.  The need for approximation techniques grow with the size of the data sets and most of the time they shed a light to make fast decisions for the businesses. General methods and techniques for handling complex tasks have room to improve in both MapReduce systems and parallel databases.  For example, consider a web site, such as a search engine, consists of several web server hosts; user queries (requests) are collectively handled by these servers (using some scheduling protocol); and the overall performance of the web site is characterized by the latency (delay) encountered by the users. The distribution of the latency values is typically very skewed, and a common practice is to track some particular quantiles, for instance, the 95th percentile latency. In this context, one can ask the following questions.
\begin{itemize}
\item What is the 95th percentile latency of a single web server for the last month?
\item What is the 95th percentile latency of the entire web site (over all the servers) for Christmas seasons?
\item How latency affected the user behavior for the holiday seasons?
\end{itemize}
The Yahoo website, for instance, handles more than 43 million hits per day~\cite{yahoo2016}, which translates to 40000 requests per second. The Wikipedia website handles 30000 requests per second at peak, with 350 web servers. While all three questions relate to computing of statistics over data, they have different technical nuances, and often require different algorithmic approaches as accuracy can be traded for performance.

One way to obtain statistical information from data is histogram construction. Histograms summarize the whole data and give information about distribution of the data. Moreover, the importance of histogram increases when the size of the data is huge. Since, histograms are very useful and are efficient ways to get quick information about data distribution, they are highly used in database systems for query optimization, query result estimation, approximate query answering, data mining, distribution fitting, and parallel data partitioning~\cite{ioannidis2003history}. One of the most used histogram types in database systems is the equi-depth histogram. The equi-depth histogram is constructed by finding boundaries that split the data into a predefined number of buckets containing equal number of tuples. More formally, $\beta$-bucket equi-depth histogram construction problem can be defined as follows: given a data set with $N$ tuples, find the boundary set $B = {b_{1}, b_{2}, \dots, b_{\beta-1}}$ that splits the sorted tuples into $\beta$ buckets, each of which has approximately $N/\beta$ tuples.

In this paper, we propose a framework to compute equi-depth histograms on-demand (dynamic) from the precomputed histograms of the partitioned data. In order to do so, we propose a histogram merging algorithm giving a user specified error bound on the bucket size. In particular, we merge $T$-bucket histograms to build a $\beta$-bucket histogram for the underling data of size $N$ and give a mathematical proof showing $2\beta/T$ error rate on the bucket size and as well as any range on the histogram. In our framework, users specify $T$ and $\beta$, we compute $T$-bucket histograms for each partition, and a query asking for a histogram of any subset of the partitions. Then, the framework computes the $\beta$ histogram on-demand from the offline computed histograms of the partitions. In real life systems, the precomputation is done incrementally (i.e., daily, hourly or monthly) such as logs and database transactions.

Our contribution can be summarized as follows:

\begin{itemize}
\item We proposed a novel algorithm to build an approximate equi-depth histogram for a union of partitions from the sub-histograms of the partitions.
\item We theoretically and experimentally showed that the error percentage of a bucket size is bounded by a predefined error set by the user (i.e., $\varepsilon_{max}$).
\item We theoretically and experimentally showed that the error percentage of a range size is bounded by a predefined error set by the user (i.e., the same above $\varepsilon_{max}$).
\item We implemented our algorithm on Hadoop and demonstrated how to apply it to practical real life problems.
\end{itemize}

The rest of the paper is organized as follows: In Section~\ref{sec:prob}, we introduce the histogram construction problem for big data. We give some background information about cloud computing environments in Section~\ref{sec:back}. In Section~\ref{sec:algo}, in-depth explanation of the proposed method takes place. The details of implementation on Hadoop MapReduce framework is given in Section~\ref{sec:impl}. In Section~\ref{sec:rel}, related works are summarized. Evaluation methodology and experimental results are discussed in Section~\ref{sec:exp} and finally we conclude the paper with Section~\ref{sec:conc}.

\section{Problem Definition}
\label{sec:prob}
\begin{figure*}[!htb]
\centering
\includegraphics[width=0.6\textwidth]{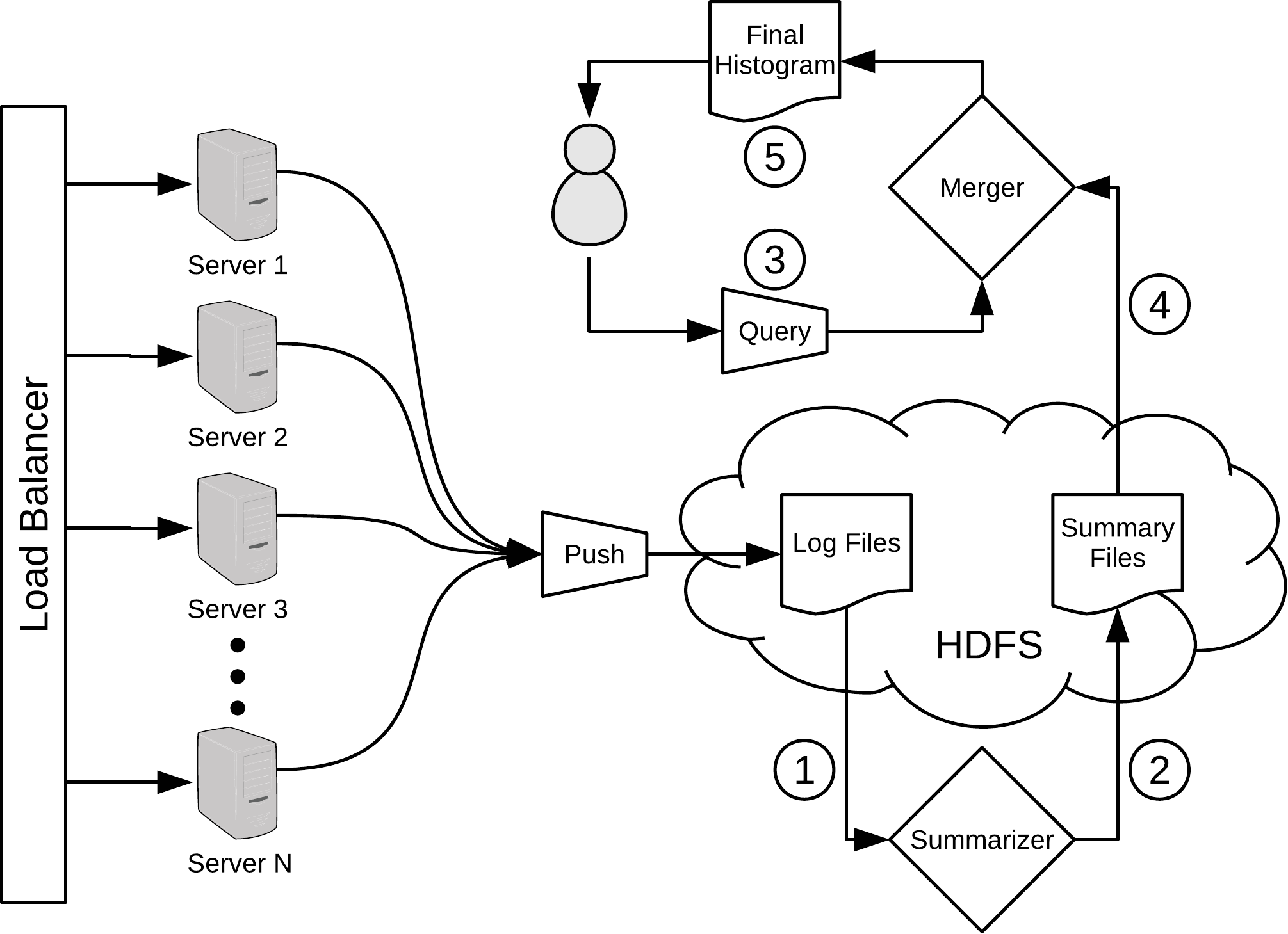}
\caption{Flowchart of the proposed method}
\label{fig:query}
\end{figure*}

\begin{figure}[!htb]
\centering
\includegraphics[width=0.5\textwidth]{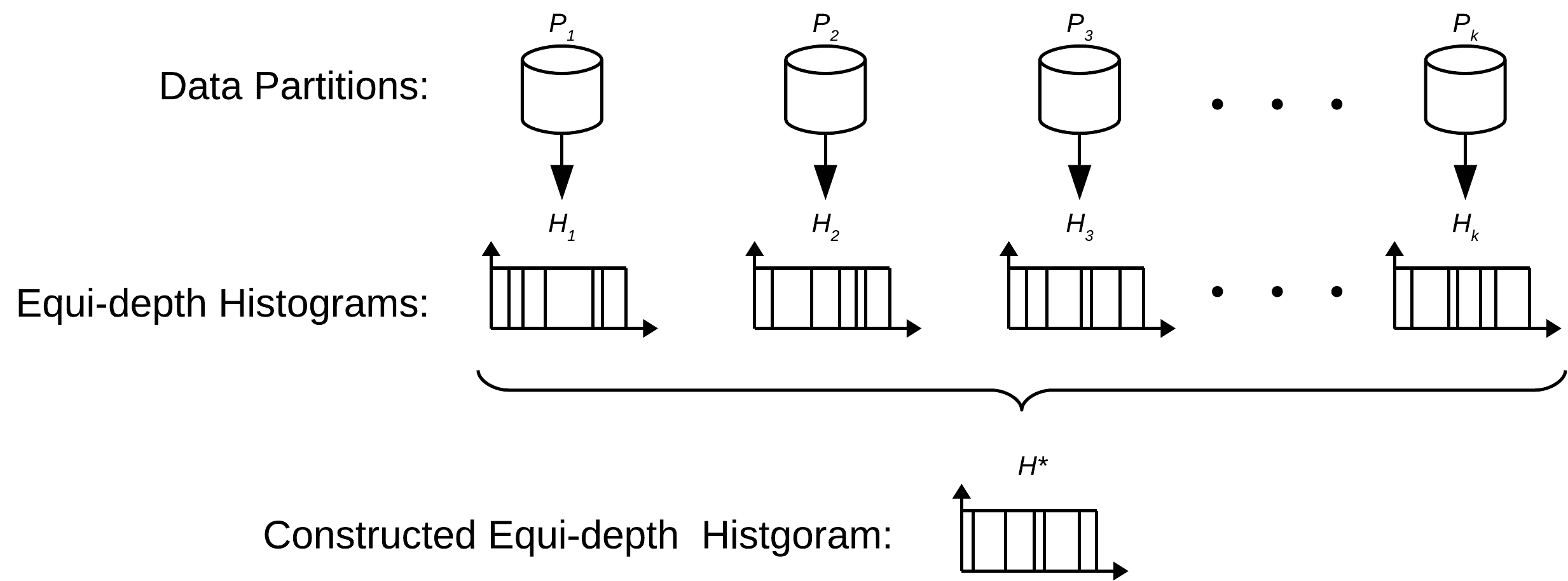}
\caption{Histogram building by merging exact histograms of data partitions}
\label{fig:merge}
\end{figure}

In this section, we motivate the problem with a practical example and then formally define it. Machine generated data, also known as logs, is automatically created by machines. Logs contain list of activities of machines. In general, logs are stored daily in W3C format~\cite{hallam1996extended}. When we consider a web server, requests from web applications and responses from the server are written to log files. There are several actors deriving intelligence from these logs. For example; operations engineers derive operational intelligence (response times, errors, exceptions etc.) and business analyst derives business intelligence (top requested pages, top users, click through rates etc.). In the context of web applications, the need to analyze clickstreams has increased rapidly, and in order to answer the demand, businesses build log management and analytical tools. A typical internet business may have thousands of web servers logging every activity on the site. In addition, they have ETL processes incrementally collecting, cleaning and storing the logs in a big data storage (i.e., This is usually Hadoop and its storage HDFS). This work-flow is demonstrated in Figure \ref{fig:query}. The amount of data to ETL and to run analytics on is huge and has been increasing rapidly. Most of the time, customers would be happy to trade accuracy for performance as they need a quick intelligence to make fast decisions.

One quick and reliable way to understand the statistics about the underlying data is using equi-depth histograms. In the paper, we outline a framework computing on-demand histograms of the data for any time interval for the above scenario. In web servers, daily logs are kept instantly. At the end of the day, all the log files belonging to that day are concatenated in a single log file and it is pushed to HDFS. As soon as the new log file is available in the HDFS, an exact equi-depth histogram is built and stored in the HDFS in a new summary file by the Summarizer Job. This means that the equi-depth histogram of each daily log is stored. Then, if a histogram for any time interval is requested (for example histogram for the last month), the Merger Job fetches an equi-depth histogram of each histogram and merges them using the proposed merging algorithm explained in the following sections. We also provide an error rate on the histogram in order to increase the confidence. Although we motivate our framework with logs, it can be applied without loss of generality to any structured data where we need a histogram such as database transactions, etc...

After motivating and showing the need, we can formulate the problem we are solving as follows: \\
\\
\textbf{\textit{Problem Definition:}} Given $k$ partitions, $P_1, P_2,..., P_k$, and their respective $T$-bucket equi-depth histograms, $H_1, H_2,...,H_k$, build a $\beta$-bucket equi-depth histogram $H^{*}$ where $\beta \leq T$ over $P_1, P_2,..., P_k$ where $ |P_1| + |P_2| + ... + |P_k|$ is equal to $N$ and $B_1, B_2,...,B_\beta$ are the buckets of $H^{*}$ such that:
\begin{itemize}
\item The size of any bucket $B_i$ is $(N/\beta) \pm \varepsilon_{max}$ where $\varepsilon_{max} < 2\beta/T \times (N/\beta)$. 
\item The size of any range spanning $m$ buckets $B_i$ through $B_j$ is $m \times (N/\beta) \pm \varepsilon_{max}$ where $\varepsilon_{max} < 2\beta/T \times (N/\beta)$. 
\end{itemize}

\section{Background}
\label{sec:back}
In this section, some background information is given about cloud computing environments such as Distributed File System (DFS)~\cite{ghemawat2003google},  MapReduce (MR)~\cite{dean2008mapreduce}, Pig~\cite{olston2008pig}, and Hive~\cite{thusoo2009hive}.

\begin{figure*}[!t]
\centering
\includegraphics[width=0.7\textwidth]{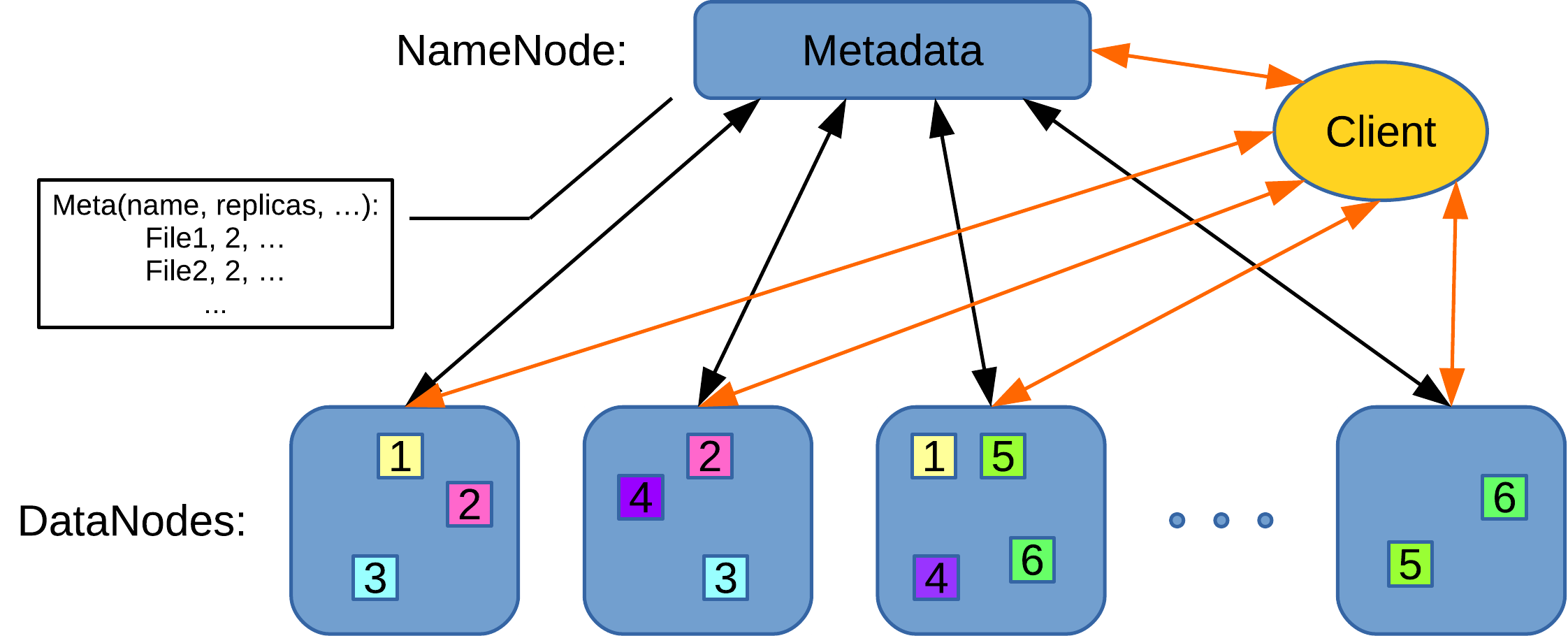}
\caption{General architecture of Hadoop Distributed File System}
\label{fig:dfs}
\end{figure*}

\textbf{\textit{Distributed File System (DFS):}} DFS is a generalized name of distributed, scalable, and fault-tolerant file systems such as Google FS~\cite{ghemawat2003google} and Hadoop DFS~\cite{borthakur2008hdfs}. In particular, we address the HDFS in this paper. In HDFS, large files are divided into small chunks and these small chunks are replicated and stored in multiple machines named DataNodes. The replication process ensures that HDFS is fault-tolerant. The metadata of the stored files such as name, replication count, file chunk locations, etc. are indexed in NameNode which is another machine. Clients read and write files to HDFS by interacting with the NameNode and the DataNodes.

The overall system architecture of HDFS is seen in Figure~\ref{fig:dfs}. In the figure, the NameNode takes place at the top and the DataNodes at the bottom. The replicas of the file chunks are labeled with the same numbers. The NameNode can interact with the DataNodes to maintain the file system by controlling the health and balancing the loads of the DataNodes. If there is a problem in a DataNode, the NameNode detects the problematic DataNode and replicates the file chunks in that DataNode to other DataNodes. Clients can also interact with the whole HDFS to read existing files and write new files to HDFS.

\begin{figure*}[!t]
\centering
\includegraphics[width=0.7\textwidth]{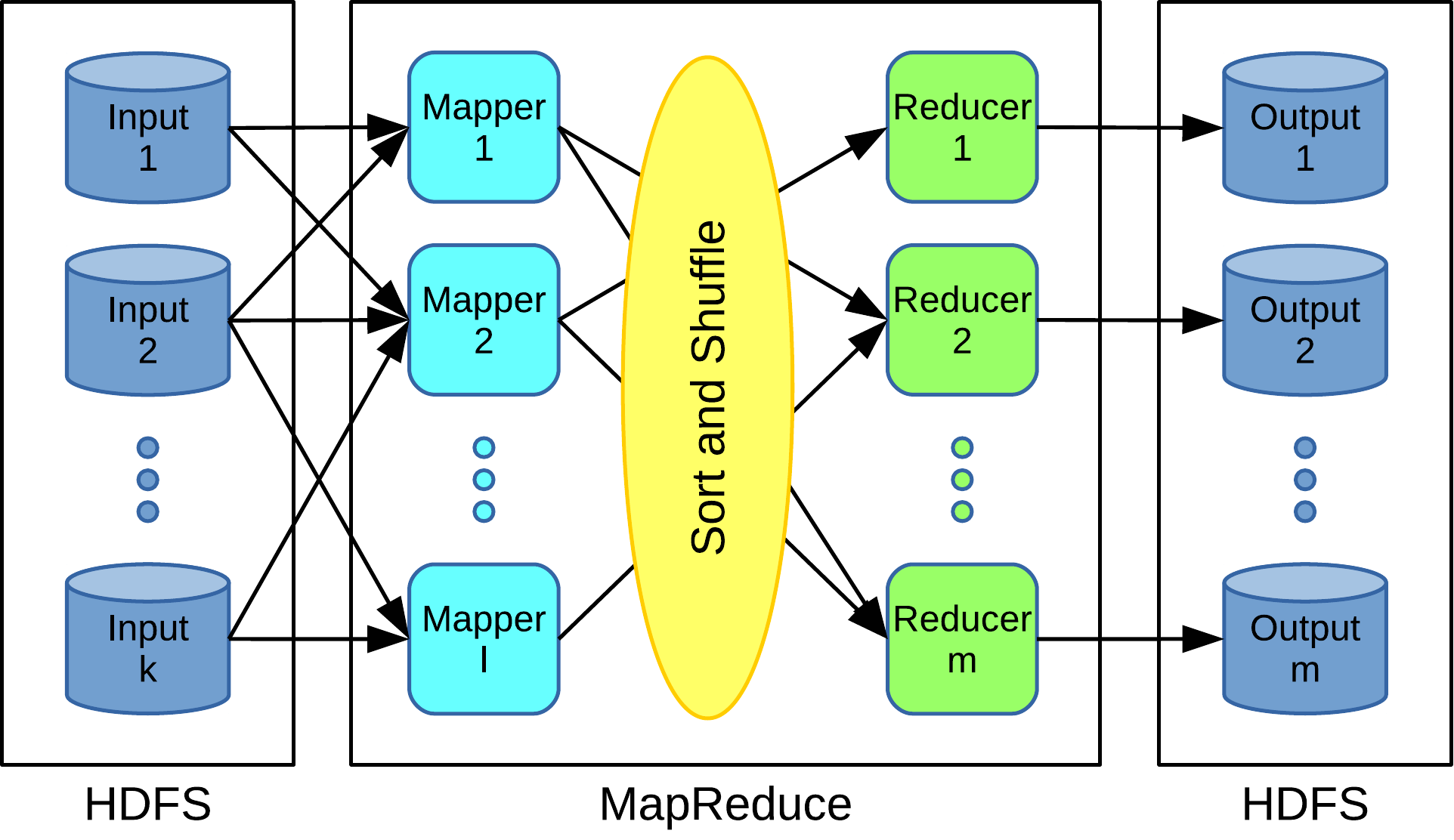}
\caption{Overview of Hadoop MapReduce Framework}
\label{fig:mr}
\end{figure*}

\textbf{\textit{MapReduce (MR):}} MapReduce is a programming model for processing huge datasets which is especially resided in distributed file systems. Besides, MapReduce framework is the combination of the components which executes submitted MapReduce tasks by managing all resources and communications among the cluster while providing for fault tolerance and redundancy. In this paper, we specifically handle the Hadoop MapReduce framework.

A MapReduce task consist of Mappers and Reducers. The Mapper has a method called Map which gets $\langle key, value \rangle$ pairs as input and emits $\langle key, value \rangle$ pairs as intermediate output. The intermediate output is shuffled and sorted by a component of the MapReduce framework at the Sort and Shuffle phase and all $\langle key, value \rangle$ pairs are grouped and sorted by keys at this phase. The output of the Sort and Shuffle phase is $\langle key, [value1, value2, ...] \rangle$ pairs and this is the input of the Reduce method which is in the Reducer. After the Reduce method finishes it's job, it also emits $\langle key, value \rangle$ pairs as final output of the MapReduce task. In some cases, a Combiner is also included in MapReduce tasks which is often the same with the Reducer. The Combiner has a Combine method which combines the output of the Map method to decrease network traffic.

The summary picture of the MapRecude framework is given in Figure~\ref{fig:mr}. In the figure, input to a Mapper is read from HDFS. The output of the Mappers goes through the Sort and Shuffle phase and Reducers get the sorted and shuffled data and process it and write the output to HDFS, again.

\textbf{\textit{Apache Pig:}} Pig is a platform for processing big data with query programs written in a procedural language called Pig Latin. Query programs are translated into MapRecude tasks and the tasks are run over MapReduce framework. The queries can be written by using both existing and user defined functions. Thus, Pig is an extensible platform and users can create their own functions.

\textbf{\textit{Apache Hive:}} Hive is another platform for storing and processing large datasets like Pig. Hive has its own SQL-like declarative querying language named as HiveQL. HiveQL also supports custom user defined Map/Reduce tasks in queries.

\section{Equi-depth Histogram Building}
\label{sec:algo}
In this section, we explain our approximate equi-depth histogram construction method in detail. In the first part of the method, exact-equi depth histograms of data partitions are constructed. This part is done offline with a well-know straight-forward histgoram construction algorithm. In the second and the important part of the method, equi-depth histograms are merged to construct an approximate equi-depth histogram over the partitions. One important feature is that the constructed histogram comes with maximum error bound on both size each bucket and size of any bucket range. 

In the following part of the section, merging part of the method is explained with an example and then the algorithm of the merging is given and the section is concluded with maximum error bound theorems and their proofs.

A $T$-bucket equi-depth histogram $H$ for a set of values $P$ (may be called a partition) can be described as an increasing sequence of numbers, which represents the boundaries. Each pair of consecutive boundaries defines a bucket, and the size of this bucket is the number of values between its boundaries, where inclusive at the front and exclusive at the end (except the last bucket). Last bucket size also includes the last boundary. For an exact equi-depth histogram, size of each bucket is the same and equals and exactly total number of values divided by total number of buckets. On the other hand, bucket sizes of an approximate equi-depth histograms may not be equal.

We express a $T$-bucket equi-depth histogram as $H=\{(b_1,s_1),(b_2,s_2), \dots , (b_i,s_i), \dots ,(b_{T-1},s_{T-1}),(b_T,0)\}$, where $b_i$ indicates the $i_{th}$ boundary and the $s_i$ indicates the $i^{th}$ bucket size for exact histograms (the approximate size of the $i^{th}$ bucket for approximate histograms), for the rest of the paper. Let us have two example value sets, $P_1$ and $P_2$, which are $\{2,4,5,6,7,10,13,16,18,20,21,25\}$ and $\{3,9,11,12,14,15,17,19,22,23,24,26,27,29,30\}$. According to the value sets, $|P_1|$ and $|P_2|$ which represent number of values in each set, equal to $12$ and $15$, respectively. $3$-bucket histogram of $P_1$ is $H_1=\{(2,4),(7,4),(18,4),(25,0)\}$ and $P_2$ is $H_2=\{(3,5),(15,5),(24,5),(30,0)\}$ and graphical representation of them are given in Figures~\ref{fig:hist1} and \ref{fig:hist2}. First bucket of $H_1$ contains the first four values, $\{2,4,5,6\}$, the second bucket contains four values, $\{7,10,13,16\}$, and the third (also the last) bucket contains the last four values $\{18,20,21,25\}$. For $H_2$, first bucket has five values, $\{3,9,11,12,14\}$, the second bucket contains five values, $\{15,17,19,22,23\}$, and the last bucket has five values, $\{24,26,27,29,30\}$.
Let us define a $s(i,H)$ function which denotes the size of $i^{th}$ bucket of the equi-depth histogram $H$, and a $S(i,H)$ function which denotes the cumulative size of all buckets from the first to $i^{th}$ bucket of $H$, that is,
\begin{equation}
\label{eqn:exactcum}
S(i,H) = s(1,H) + s(2,H) + \cdots + s(i,H)
\end{equation}
Then, the convention assures that $S(i,H) = i \times |P|/T$, for all $i \leq T$, where $|P|$ is the number of values and $T$ is the number of buckets. Considering $H_1$, $s(1,H_1)=s(2,H_1)=s(3,H_1)=4$. For cumulative sizes, $S(1,H_1)=4$, $S(2,H_1)=8$, and $S(3,H_1)=12$. Bucket sizes of $H_2$ is $s(1,H_2)=s(2,H_2)=s(3,H_2)=5$ and cumulative sizes are $S(1,H_2)=5$, $S(2,H_2)=5$, and $S(3,H_2)=5$.

Let us define two more functions, $a(i,H)$ and $A(i,H)$, which are the $i^{th}$ approximate bucket size and the $i^{th}$ cumulative bucket size for approximate equi-depth histograms, respectively. By writing an approximate version of Equation~\ref{eqn:exactcum}, we get the following equation:
\begin{equation}
\label{eqn:appxcum}
A(i,H) = a(1, H) + a(2,H) + \cdots + a(i,H)
\end{equation}
Lastly, let us define a range function $R(i,j,H)$ that gives the sum of sizes of buckets which starts from the $i^{th}$ bucket up to the $j^{th}$ bucket, both inclusive. The formal definitions are given in the following formulas for both exact bucket sizes and approximate bucket sizes.
\begin{eqnarray}
\label{eqn:rangecum}
R_{s}(i,j,H) = s(i, H) + s(i+1,H) + \cdots + s(j,H) \\
R_{a}(i,j,H) = a(i, H) + a(i+1,H) + \cdots + a(j,H)
\end{eqnarray}

The definitions given belove are summarized in Table~\ref{tbl:definitions}.

\begin{figure}[!t]
\centering
\includegraphics[width=0.487\textwidth]{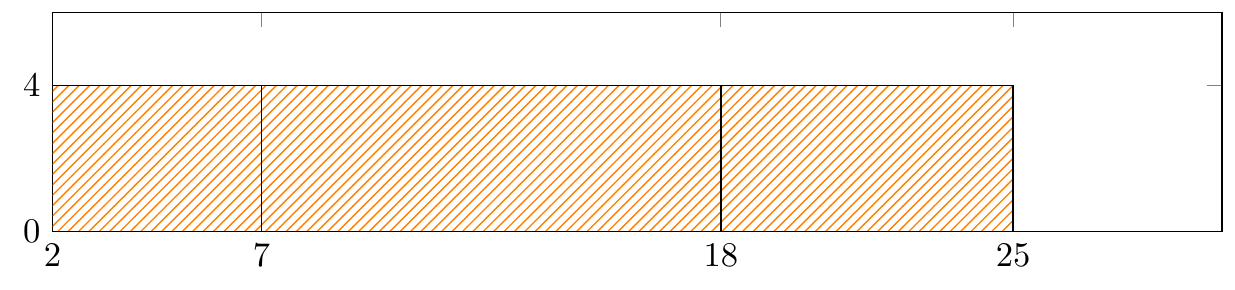}
\caption{A sample equi-depth histogram $H_1$ with $3$ buckets, based on data $\{2,4,5,6,7,10,13,16,18,20,21,25\}$.}
\label{fig:hist1}
\end{figure}

\begin{figure}[!t]
\centering
\includegraphics[width=0.5\textwidth]{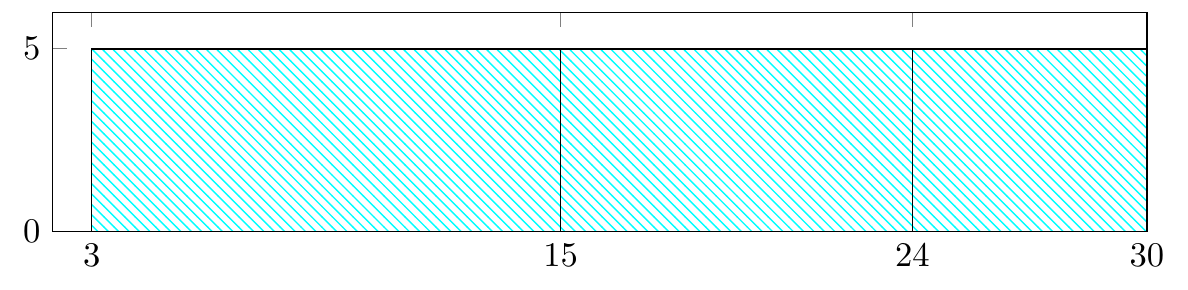}
\caption{Another sample equi-depth histogram $H_2$ with $3$ buckets, which represents data $\{3,9,11,12,14,15,17,19,22,23,24,26,27,29,30\}$.}
\label{fig:hist2}
\end{figure}

Since we completed the definitions for convention, we start to explain the merging process in detail. We have exact $3$-bucket equi-depth histograms $H_1$ and $H_2$ given in Figures~\ref{fig:hist1} and~\ref{fig:hist2} for the example value sets $P_1$ and $P_2$ where $P_1=\{2,4,5,6,7,10,13,16,18,20,21,25\}$ and $P_2=\{3,9,11,12,14,15,17,19,22,23,24,26,27,29,30\}$. The total number of values $N$ is equal to $|P_1| + |P_2| = 12 + 15 = 27$ and let bucket count of final histogram, $\beta$, be 3. As seen in the histograms $H_1$ and $H_2$, $H_1$ has a boundary sequence of $2,7,18,25$ and each $H_1$ bucket has $12/3=4$ values, $H_2$ has $3,15,24,30$ and bucket size of $15/3=5$. In Figure~\ref{fig:histsep}, we show $H_1$ and $H_2$ on the same plot, so therein, we clearly see the overall sequence of boundary values, which is $2,3,7,15,18,24,25,30$. Although the desired final number of buckets $\beta$ may be chosen to be any number less than or equal to 3, we drive the example merging for $\beta=3$.

\begin{figure}[!t]
\centering
\includegraphics[width=0.5\textwidth]{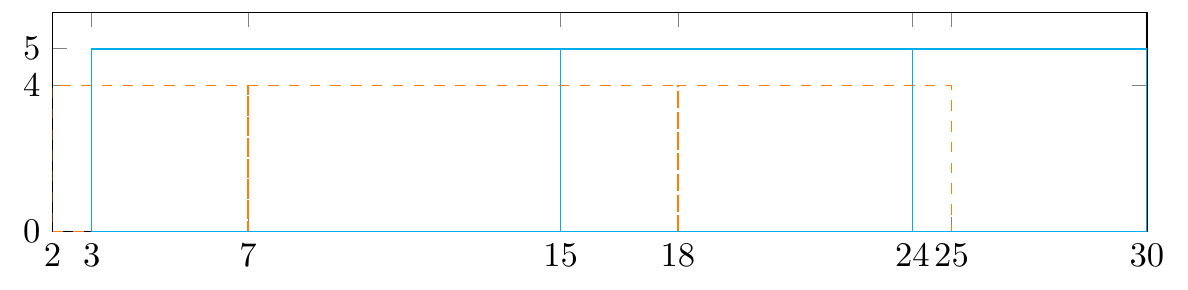}
\caption{Example equi-depth histograms given in Figures~\ref{fig:hist1} (orange-dashed) and~\ref{fig:hist2} (cyan-solid) are coupled together, with boundary sequence $\{2,3,7,15,18,24,25,30\}$.}
\label{fig:histsep}
\end{figure}

\begin{table}[!t]
\renewcommand{\arraystretch}{1.3}
\caption{Symbol Table Used in Section~\ref{sec:algo}}
\label{tbl:definitions}
\centering
\begin{tabular}{|c|l|}
\hline
Symbol & Definition\\
\hline
$H^0$ & \begin{tabular}[c]{@{}l@{}}Pre-histogram after first assembling of $k$ exact\\ equi-depth histograms, $H_1$, $H_2$, $\ldots$, and $H_k$\end{tabular}\\ \hline
$H^*$ & \begin{tabular}[c]{@{}l@{}}Final approximate equi-depth histogram after\\ bucket merging operations of $H^0$\end{tabular}\\ \hline
$H^e$ & \begin{tabular}[c]{@{}l@{}}Exact equi-depth histogram for the union of $k$\\ value sets, $P_1$, $P_2$, $\ldots$, and $P_k$\end{tabular}\\ \hline
$s(i,H)$ & \begin{tabular}[c]{@{}l@{}}The $i^{th}$ bucket size of the equi-depth histogram $H$\end{tabular}\\ \hline
$a(i,H)$ & \begin{tabular}[c]{@{}l@{}}The $i^{th}$ approximate bucket size of the approximate\\ equi-depth histogram $H$\end{tabular}\\ \hline
$S(i,H)$ & \begin{tabular}[c]{@{}l@{}}The $i^{th}$ cumulative size of the equi-depth\\ histogram $H$\end{tabular}\\ \hline
$A(i,H)$ & \begin{tabular}[c]{@{}l@{}}The $i^{th}$ approximate cumulative size of the\\ equi-depth histogram $H$\end{tabular}\\ \hline
$R(i,j,H)$ & \begin{tabular}[c]{@{}l@{}}Sum of bucket sizes starting from the $i^{th}$ bucket\\ up to the $j^{th}$ bucket (both inclusive) of the\\ equi-depth histogram $H$\end{tabular}\\ \hline
\end{tabular}
\end{table}

Let us name the calculated pre-histogram (after first assembling of $H_1$ and $H_2$) as $H^0$, final merged approximate equi-depth histogram resulted from our method as $H^*$, and exact equi-depth histogram for the union of value sets $P_1$ and $P_2$ as $H^e$. Briefly, we start with assembling $H_1$ and $H_2$ in an initial pre-histogram $H^0$ and we merge consecutive buckets of $H^0$ while the merged bucket size is greater than or equal to the exact bucket size $N/\beta$ till the remaining number of buckets is equal to the desired number $\beta$. For $\beta=3$, exact bucket sizes $N/\beta$ should be equal to $27/3=9$ and it can be presented as $s(1,H^e) = s(2,H^e) = s(3,H^e) = 9$. The cumulative bucket sizes for $H^e$ are $S(1,H^e) = 9$, $S(2,H^e) = 18$, and $S(3,H^e) = 27$. Now, we shall examine the creation of the pre-histogram $H^0$. The boundaries of $H^0$ are $2,3,7,15,18,24,25,30$ which are shortly the sorted boundaries of $H_1$ and $H_2$. Since $H^0$ has $(T + 1) \times 2 = (3 + 1) \times 2 = 8$ boundaries, it has $8 - 1 = 7$ buckets. The important part of the creation is approximation of bucket sizes of $H^0$. Before the approximation of bucket sizes, we should determine approximate cumulative bucket sizes of $H^0$ and then we are able to calculate the approximate bucket sizes from the definition of the cumulative bucket size function $A(i,H)$. The approximate cumulative bucket sizes are calculated by presuming that all values in each bucket are at the beginning boundary of the bucket. For example, let us consider the first bucket of $H_1$. In this bucket, we have values $2$, $4$, $5$, and $6$ and we suppose that all these values are at the point $2$. By using this supposition, any cumulative bucket size is easily determined by summing the bucket size of the histogram which holds the next boundary and the previous cumulative bucket size starting with $0$. Thus, since the first boundary of $H^0$ is $2$ and this boundary is the first boundary of $H_1$, the first cumulative approximate bucket size $A(1,H^0)$ is equal to $s(1,H_1) = 4$. The second cumulative approximate bucket size $A(2,H^0)$ is equal to $s(1,H_2) + A(1,H^0) = 5 + 4 = 9$ because the next coming $H^0$ boundary is $3$ and it is the first boundary of $H_2$. After the boundary $3$, the next $H^0$ boundary is $7$ and it is the second boundary of $H_1$. Therefore, the third cumulative approximate bucket size $A(3,H^0)$ is equal to $s(2,H_1) + A(2,H^0) = 4 + 9 = 13$. The remaining approximate bucket sizes are calculated in the same way and $A(4,H^0)$, $A(5,H^0)$, $A(6,H^0)$, and $A(7,H^0)$ are $18$, $22$, $27$, and $27$, respectively. Now, we are able to calculate approximate bucket sizes. Approximate size of the first bucket $a(1,H^0)$ relying between boundary $2$ and $3$ is directly equal to $A(1,H^0)$ and it is $4$. The second approximate bucket size $a(2,H^0)$ is the difference between the first and the second cumulative bucket size. Thus, $a(2,H^0)$ is equal to $A(2,H^0) - A(1,H^0) = 9 - 4 = 5$. Similarly, $a(3,H^0) = A(3,H^0) - A(2,H^0) = 13 - 9 = 4$, $a(4,H^0) = 5$, $a(5,H^0) = 4$, $a(6,H^0) = 5$, and $a(7,H^0) = 0$. Graphical representation of created $H^0$ is given in Figure~\ref{fig:histmerge}.

\begin{figure}[!t]
\centering
\includegraphics[width=0.5\textwidth]{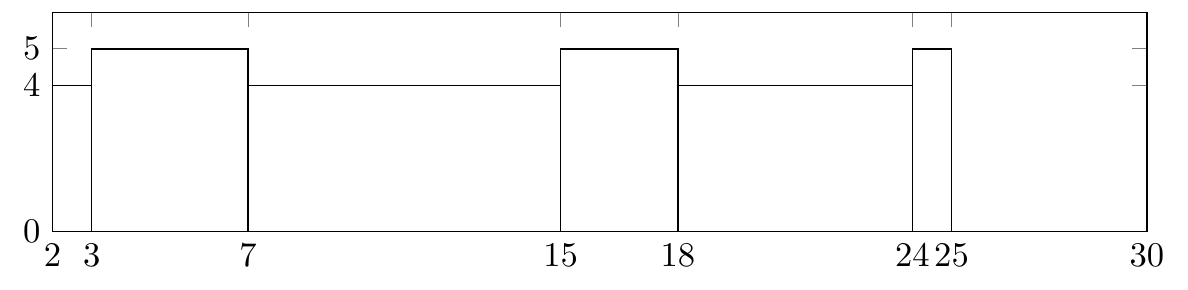}
\caption{The initial pre-histogram $H^0$ constructed just after the first assembling of $H_1$ and $H_2$.}
\label{fig:histmerge}
\end{figure}

Next, we merge the buckets of $H^0$ until the remaining bucket count is equal to $\beta$. We use approximate cumulative bucket size instead of approximate bucket size to decrease the division error while merging. The merging process starts with the first bucket of $H^0$. First of all we compare the first cumulative bucket size of $H^0$, $A(1,H^0)$, with the first cumulative bucket size of exact (ideal) histogram, $S(1,H^e)$, and we see that $A(1,H^0)$ is less than $S(1,H^e)$. We continue comparing the next cumulative bucket size of $H^0$, $A(2,H^0)$, with again the first cumulative bucket size of exact (ideal) histogram, $S(1,H^e)$, and we now see that $A(2,H^0)$ is equal to $S(1,H^e)$. Again, we continue comparing. This time, we see that $A(3,H^0)$ is greater than $S(1,H^e)$.Therefore, the buckets starting from the first bucket to the third bucket except the third one (because the result of the previous comparison is equality) would be merged and this merged bucket would be the first bucket of the final merged approximate histogram, $H^*$. The resulting new bucket size would be $A(2,H^0)$ because the new merged bucket is the first bucket of $H^*$. Then, we are going to create the second bucket of $H^*$. For this creation, we continue comparing cumulative bucket sizes starting from the first not merged bucket number with the second cumulative bucket size of $H^e$. We see that $A(3,H^0)$ is less than $S(2,H^e)$. Next comparison is between the next cumulative of $H^0$ and again the second cumulative of $H^e$. This time equality is seen. We continue comparing the next cumulative of $H^0$, $A(5,H^0)$ with $S(2,H^e)$. At this point, $A(5,H^0)$ is greater than $S(2,H^e)$. Thus, we merge the buckets starting from the third one to the fifth one again except the fifth one and the created new bucket would be the second bucket of $H^*$. This merging process would end when the remaining bucket count is equal to $\beta$ and we get $H^*=\{(2,9),(7,9),(18,9),(30,0)\}$ as seen in Figure~\ref{fig:histfinal}. For comparison, $H^e$ is given in Figure~\ref{fig:histexact}.

\begin{figure}[!t]
\centering
\subfloat[The final approximate histogram $H^*$ constructed by merging $H_1$ and $H_2$.]{
    \centering
    \includegraphics[width=0.5\textwidth]{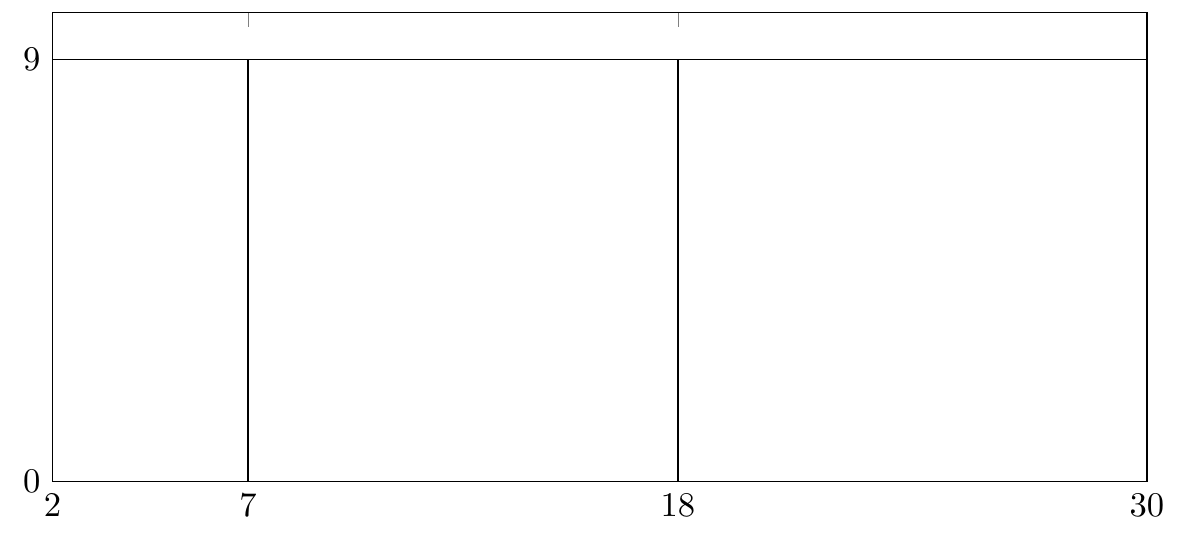}%
    \label{fig:histfinal}
}
\hfil
\subfloat[The exact histogram for union of $P_1$ and $P_2$.]{
    \centering
    \includegraphics[width=0.5\textwidth]{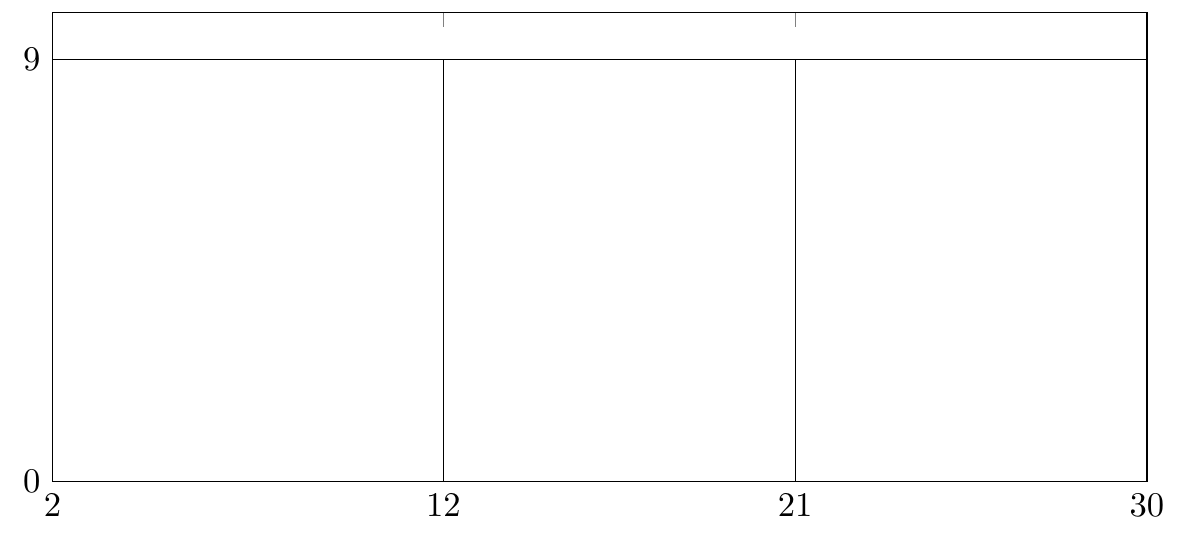}%
    \label{fig:histexact}
}
\caption{The final approximate and exact histograms of the example value sets $P_1$ and $P_2$.}
\label{fig:buildhists}
\end{figure}

The generalization of this method for merging more than 2 histograms is now easy after the one given above. Let us have k value sets ($P_1$, $P_2$, ..., $P_k$) and their summaries (T-bucket equi-depth histograms, $H_1$, $H_2$, ..., $H_k$) to be merged. The merging process for the general case starts with the creation of an initial pre-histogram, $H^0$. This can be done with sorting all boundary values coming from summaries and determining approximate bucket sizes in the same way with the one described above. The calculated histogram $H^0$ has $k \times (T + 1)$ boundaries and thus $k \times (T + 1) - 1$ buckets. The rest of the merging method is exactly the same with the case when we have only two histograms. That is, we combine consecutive buckets of $H^0$ by comparing the cumulative bucket sizes of $H^0$ with cumulative sizes of exact histogram, $H^e$, until $\beta$ buckets remain.

\begin{algorithm}[t]
\KwIn{$H_1, H_2, \ldots, H_k$: $k$ equi-depth histograms each with $T$ buckets, $N$: total number of values, $\beta$: desired bucket count of final histogram}
\KwOut{$H^*$: an approximate equi-depth histogram with $\beta$ buckets.}
$b \leftarrow$ \{sorted boundaries of $H_1, H_2, \ldots, H_k$\} \label{alg:line:b} \\
$s \leftarrow$ \{bucket sizes calculated as described\} \label{alg:line:s} \\
$H^0 \leftarrow$ {\sc CreateHistogram}(b, s) \label{alg:line:h0} \\
$H^* \leftarrow H^0$ \label{alg:line:copy} \\
$last \leftarrow 1$; $next \leftarrow 1$; $current \leftarrow 1$ \\
$\rm{remaining} \leftarrow k(T+1)-1$ \\
\While{$\rm{remaining} > \beta$} { \label{alg:line:loop:begin}

  \While{$A(next,H^0) \leq current \times N / \beta$} { \label{alg:line:iloop:begin} 
      $next \leftarrow next + 1$ \label{alg:line:iloop:end} \\
  }
  {\sc MergeBuckets}($last, next - 1, H^*$) \label{alg:line:combine} \\
  $last \leftarrow next$; $current \leftarrow current + 1$ \\
  $\rm{remaining} \leftarrow \rm{remaining} - (next - 1 - last)$  \label{alg:line:loop:end}
}
\Return{$H^*$}
\caption{Equi-depth Histogram Merging}
\label{alg:merge}
\end{algorithm}

Algorithm~\ref{alg:merge} shows the pseudocode of the explained method above. The algorithm takes $T$-bucket equi-depth histograms of $k$ value sets, total number of values, $N$, which is the sum of all sizes of value sets, and desired bucket count of final histogram, $\beta$ as inputs and constructs and returns final approximate $\beta$-bucket equi-depth histogram, $H^*$. Lines~\ref{alg:line:b} through~\ref{alg:line:h0} of the algorithm is performed for the creation of the initial pre-histogram, $H^0$. First, boundaries of input histograms are sorted at Line~\ref{alg:line:b} and then bucket sizes are calculated according to the above example at Line~\ref{alg:line:s}. The subroutine {\sc CreateHistogram} called at Line~\ref{alg:line:h0} simply creates a histogram from given boundary and bucket size sets and at that line $H^0$ is created from $b$ and $s$. The created $H^0$ has $k \times (T + 1)$ boundaries and $k \times (T + 1) - 1$ buckets. After creation of $H^0$, it would be copied to $H^*$ at Line~\ref{alg:line:copy}. Once $H^0$ is created and copied to $H^*$, required buckets are combined on $H^*$ considering ideal bucket size, $N/\beta$. The main {\tt While} loop iterates until the remaining number of buckets is equal to $\beta$. The inner {\tt While} loop given in Lines~\ref{alg:line:iloop:begin} and~\ref{alg:line:iloop:end} seeks for the next feasible point of buckets to combine at each iteration of the main loop. When such a point is found, we apply {\sc MergeBuckets} subroutine which combines buckets from $last$ to $next - 1$, both inclusive, on $H^*$ as shown in Line~\ref{alg:line:combine}. Notice that {\sc MergeBuckets} merges buckets according to the first state of bucket indexes.

For the asymptotic performance of the algorithm, sorting boundaries is likewise merging $k$ sorted lists and it can be done in $O(Tk\log k)$. Bucket sizes and {\sc CreateHistogram} subroutine can both run in $O(Tk)$ at Lines~\ref{alg:line:s} and~\ref{alg:line:h0}. For the inner loop, the increment at Line~\ref{alg:line:iloop:end} can be performed at most $\beta$ times. The number of iterations for the main loop changes with the decrease in remaining bucket counts. Observe that the decrease is equal to the inner loop iteration number and {\sc MergeBuckets} subroutine takes the same time with the inner loop for each main loop iteration. Considering this observation, total time required for the main loop is $O(Tk)$. Consequently, the initial sorting dominates the rest of the algorithm, and the algorithm runs in $O(Tk\log k)$-time.

Let us debug the algorithm line by line for the two example $3$-bucket equi-depth histograms $H_1$ and $H_2$ given in Figures~\ref{fig:hist1} and \ref{fig:hist2}. Recall that $H_1$ and $H_2$ are histograms of value sets $P_1$ and $P_2$. Therefore $N$ is equal to $|P_1| + |P_2| = 12 + 15 = 27$. Let $\beta$ is equal to $3$. We know $H^0=\{(2,4),(3,5),(7,4),(15,5),(18,4),(24,5),(25,0),(30,0)\}$ from the given detailed explanation above. In addition, the start state of $H^*$ is the same as $H^0$. The variables $last$, $next$, and $current$ is equal to $1$ and $remaining$ is calculated as $k(T + 1) - 1 = 2(3 + 1) - 1 = 7$. Because the $remaining$ is greater than $\beta$ at current state, we enter the main loop. For the inner loop, $A(1, H^0)$ and $A(2, H^0)$ is less or equal to $current \times N / \beta$ which is $1 \times 27 / 3 = 9$ but $A(3, H^0)$ is greater than $9$. Hereby, inner {\tt While} loop 2 times and $next$ would be $3$. Then, {\sc MergeBuckets} subroutine merges the buckets $1$ and $2$ of $H^*$. The illustration of $H^*$ is shown in Figure~\ref{fig:histdebug} after merging. The variables $last$ and $current$ are updated after the execution of {\sc MergeBuckets} is finished and $last$ would be $3$ and $current$ would be $2$. The $remaining$ variable, keeping the remaining bucket number of $H^*$, would be $6$ after the calculation is done at Line~\ref{alg:line:loop:end}. The main loop finishes after $H^*$ has $\beta$ buckets and execution of the algorithm ends with returning the created $H^*$.

\begin{figure}[!t]
    \centering
    \includegraphics[width=0.5\textwidth]{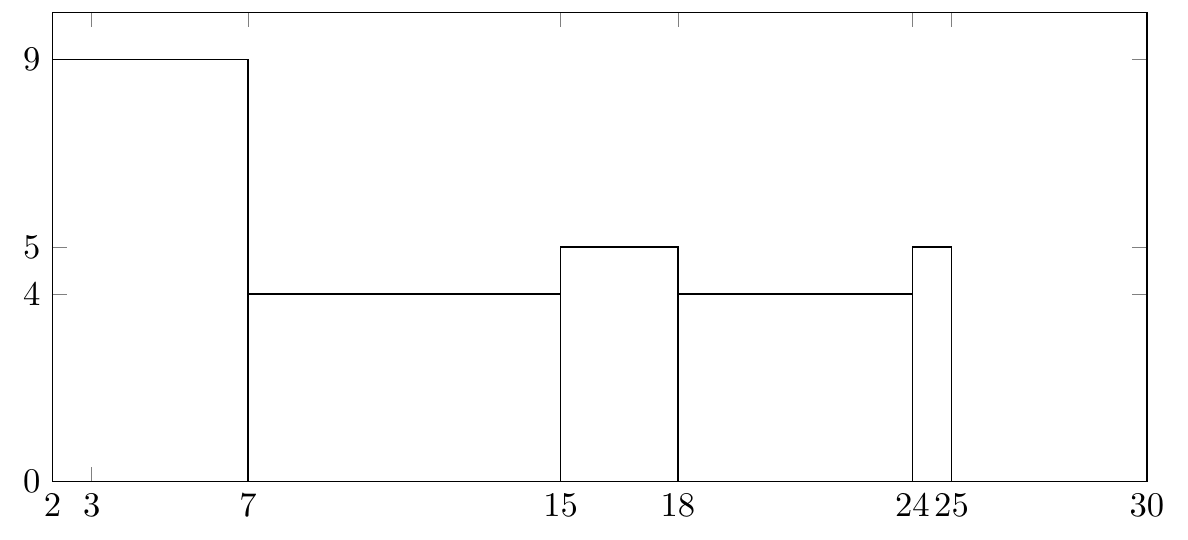}
    \caption{The state of $H^*$ after the first iteration of main loop of~\ref{alg:merge}.}
    \label{fig:histdebug}
\end{figure}

Now, we  discuss the error bounds of the output histogram $H^*$. The following two theorems and their proofs verify the error bounds on bucket sizes and the sum of any range of bucket sizes of $H^*$.
\begin{theorem}
Let $H_1$, $H_2$, \ldots, $H_k$ be $T$-bucket equi-depth histograms of value sets $P_1$, $P_2$, \ldots, $P_k$, and $H^*$ be the approximate $\beta$-bucket equi-depth histogram where $\beta \leq T$ constructed by the algorithm. Then, the size of any bucket $a(i, H^*)$ is $(N/\beta) \pm \varepsilon_{max}$ where $\varepsilon_{max} < 2\beta/T \times (N/\beta)$.
\label{thm:single}
\end{theorem}
\begin{proof}
\label{prf:single}
Recall that the calculations of bucket sizes of $H^0$ depends on supposition that all values in each bucket are at the beginning boundary of the bucket and $H^*$ is some-buckets-merged version of $H^0$. Now consider an $i^{th}$ bucket between the $i^{th}$ and the ${i+1}^{th}$ boundaries (boundaries may be any of the two consecutive boundaries of $H^0$) of $H^*$ illustrated in Figure~\ref{fig:single_proof}. As seen in the figure, all of the values in the buckets divided by the $i^{th}$ boundary may stay at the right hand side of the boundary in contrast to our assumption and all values in the buckets divided by the ${i+1}^{th}$ boundary may stay at the left hand side of the boundary. In this case, $a(i, H^*)$ gets the maximum value. Vice versa, $a(i, H^*)$ gets the minimum value in the case that all possible values in the divided buckets stay out of the $i^{th}$ bucket in contrast to the case seen in Figure~\ref{fig:single_proof}. The following calculation shows the maximum value of $a(i, H^*)$.
\begin{eqnarray}
{a(i, H^*)}_{max} & = & C + |P_1|/T + |P_3|/T \nonumber \\
&& {+}\: \cdots + |P_k|/T \nonumber \\
&& {+}\: |P_1|/T + |P_2|/T \nonumber \\
&& {+}\: \cdots + |P_{k-1}|/T \label{eqn:thm:max1}
\end{eqnarray}
where $C$ is constant which is the sum of the sizes of the buckets relying completely in the $i^{th}$ bucket and $|P_1|$, $|P_2|$, \ldots, $|P_k|$ is the size of sets. Adding and subtracting $|P_2|/T$ and $|P_k|/T$ to the equation, we get the following equation.
\begin{eqnarray}
{a(i, H^*)}_{max} & = & C + (|P_1| + |P_2| + \cdots + |P_k|)/T \nonumber \\
&& {+}\: (|P_1| + |P_2| + \cdots + |P_k|)/T \nonumber \\
&& {-}\: |P_2|/T - |P_k|/T \nonumber \\
& = & C + 2N/T - |P_2|/T - |P_k|/T \nonumber \\
& < & C + 2N/T
\end{eqnarray}
And ${a(i, H^*)}_{min}$ is equal to $C$ because no additional values are located in the $i^{th}$ bucket except the constant ones. Once ${a(i, H^*)}_{max}$ and ${a(i, H^*)}_{min}$ are determined, $\varepsilon_{max}$ would be the difference between them.
\begin{eqnarray}
\varepsilon_{max} & = & {a(i, H^*)}_{max} - {a(i, H^*)}_{min} \nonumber \\
& < & C + 2N/T - C \nonumber \\
& < & 2N/T
\label{eqn:thm:single_error_bound1}
\end{eqnarray}
The following equation shows another expression of $\varepsilon_{max}$ in terms of exact (ideal) bucket size $N/\beta$.
\begin{eqnarray}
\varepsilon_{max} & = & 2N/T \nonumber \\
& < & 2N\beta/T\beta \nonumber \\
& < & 2\beta/T \times (N/\beta)
\label{eqn:thm:single_error_bound2}
\end{eqnarray}

\begin{figure}[!t]
    \centering
    \includegraphics[width=0.5\textwidth]{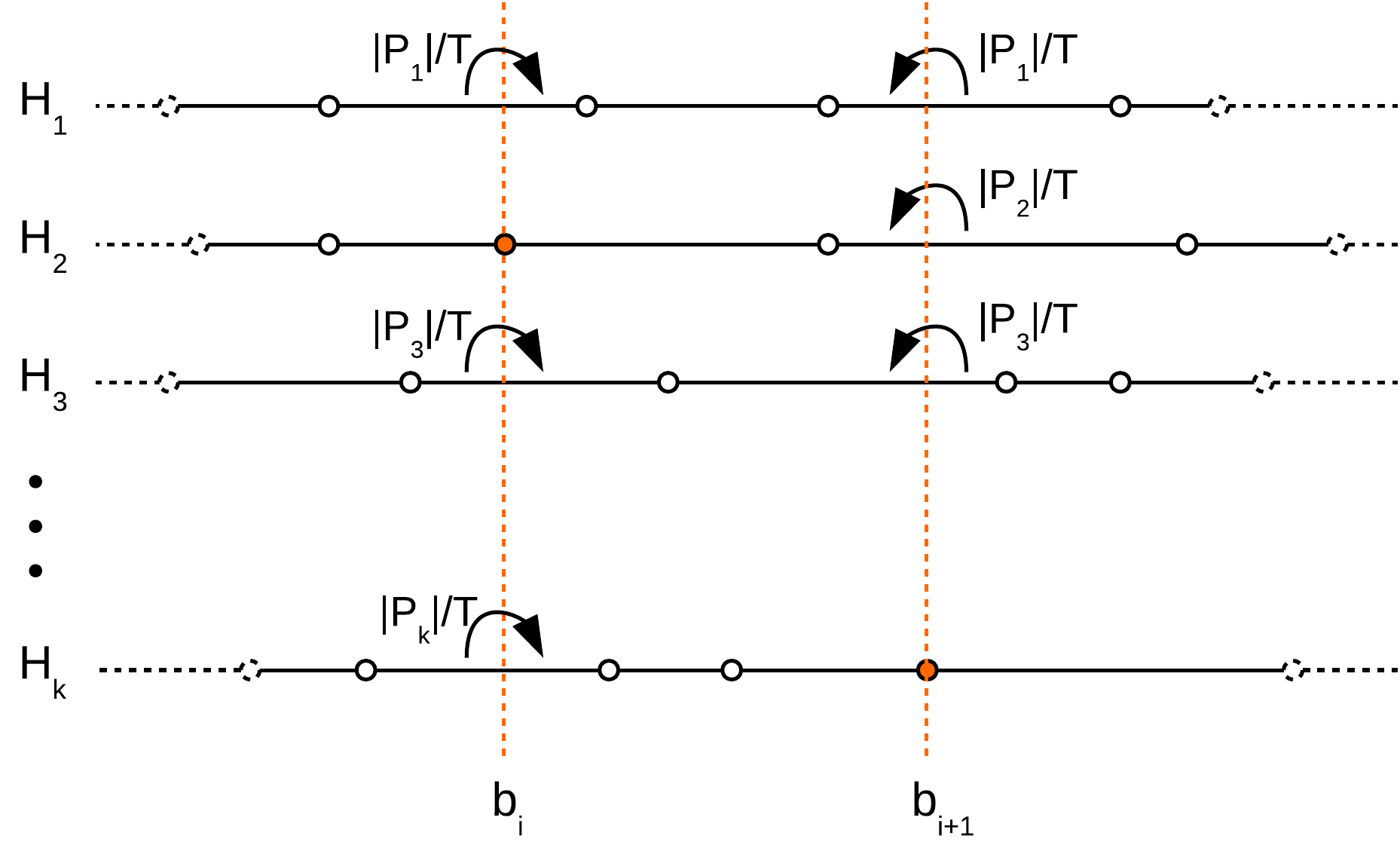}
    \caption{Illustration of maximum bucket size.}
    \label{fig:single_proof}
\end{figure}
\end{proof}

\begin{theorem}
Let $H_1$, $H_2$, \ldots, $H_k$ be $T$-bucket equi-depth histograms of value sets $P_1$, $P_2$, \ldots, $P_k$, and $H^*$ be the approximate $\beta$-bucket equi-depth histogram where $\beta \leq T$ constructed by the algorithm. Then, the size of any range spanning $m$ buckets, $R_a(i,i+m,H^*)$, is $m \times (N/\beta) \pm \varepsilon_{max}$ where $\varepsilon_{max} < 2\beta/T \times (N/\beta)$.
\label{thm:range}
\end{theorem}
\begin{proof}
Let us start with proving the error bound of range size of two consecutive buckets. Figure~\ref{fig:range_proof} shows this case. There are two consecutive buckets and three boundaries ($b_i$, $b_{i+1}$, and $b_{i+2}$), the middle one ($b_{i+1}$) splits the two buckets. Notice that the intersected buckets by $b_{i+1}$ completely rely in the range of the two buckets and this means that the sizes of these buckets are added as a constant to the range size $R_a(i,i+1,H^*)$. As a result, this proof turns into the proof of error bound of bucket size given in Proof~\ref{prf:single} and Equation~\ref{eqn:thm:single_error_bound1} and Equation~\ref{eqn:thm:single_error_bound1} also proves Theorem~\ref{thm:range}. The general case -spanning ranges includes more than two buckets- can also transform into a single bucket problem in the same way with the case with two buckets.

\begin{figure}[!t]
    \centering
    \includegraphics[width=0.5\textwidth]{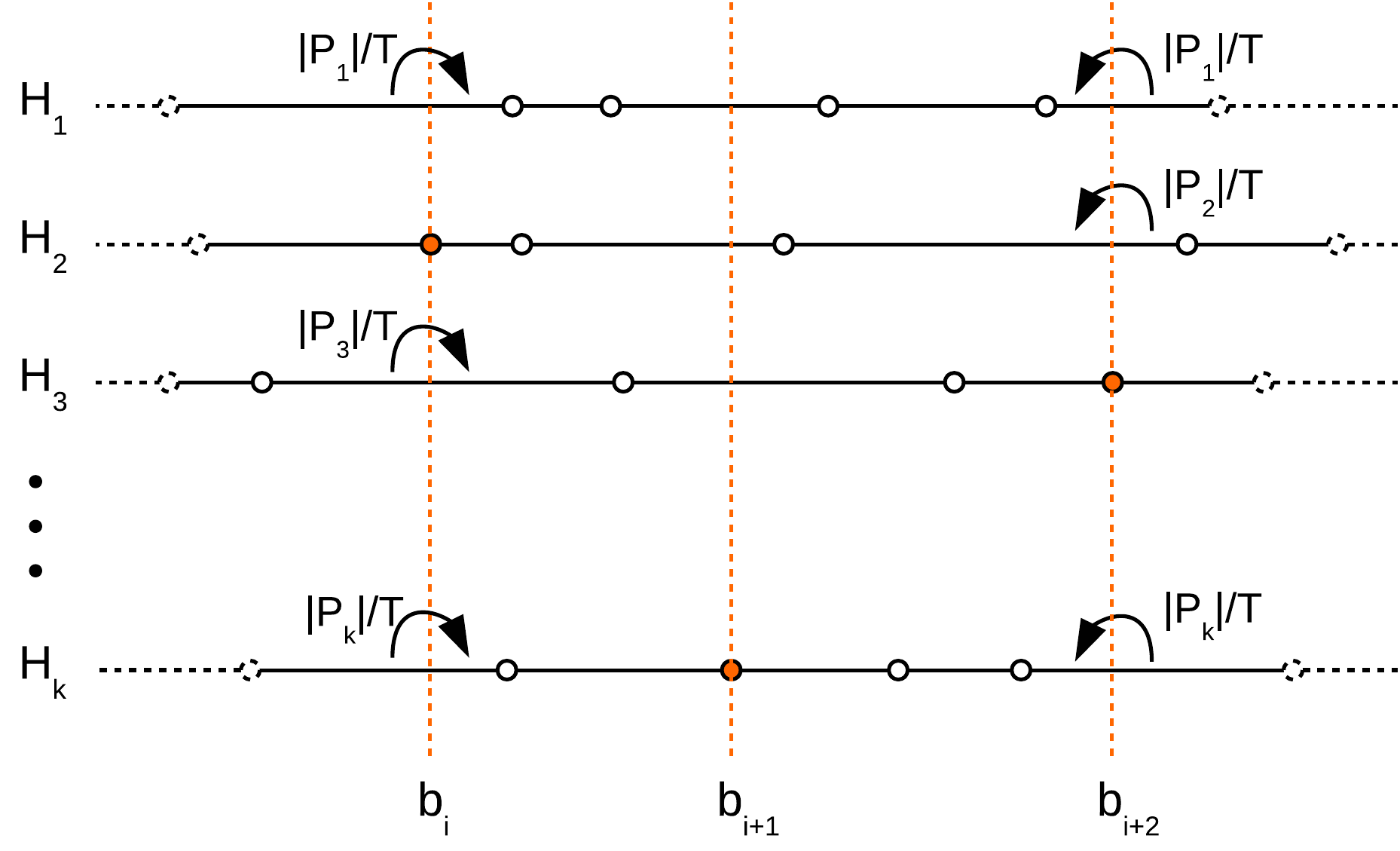}
    \caption{Illustration of maximum size of a range of buckets.}
    \label{fig:range_proof}
\end{figure}
\end{proof}

According to Theorems~\ref{thm:single} and~\ref{thm:range}, users can bound the maximum bucket size error of final $\beta$-bucket approximate equi-depth histogram $H^*$ by selecting appropriate bucket numbers $\beta$ and $T$. For example, let us calculate $T$, number of buckets of equi-depth histograms of data partitions kept in the summary files, in terms of $\beta$ for getting final merged histograms, the maximum bucket size errors of which do not exceed 5\% of the ideal bucket size ($N/\beta$). If we use Equation~\ref{eqn:thm:single_error_bound2}, we can find the minimum number of buckets $T$ needed to satisfy the 5\% error condition as follows.
$$\varepsilon_{max} < 2\beta/T \times (N/\beta) \leq 0.05(N/\beta)$$
$$40\beta \leq T$$
Consequently, the required bucket size $T$ should be at least 40 times $\beta$ which is the desired number of buckets of constructed histograms using our method.

\section{Implementation With Hadoop Map-Reduce}
\label{sec:impl}
In this section, we explain the implementation details of our histogram processing framework on Hadoop MapReduce. The framework consists of two main MapReduce jobs. One of them is named as Summarizer which runs offline and is scheduled for summarizing the new coming data to HDFS. The Summarizer constructs a $T$-bucket equi-depth histogram of the data. After summarizing, the resulting equi-depth histograms are stored in HDFS. The second job, Merger, is run on-demand according to users' requests. Its duty is to merge the related summaries from HDFS by considering user requests and to construct the final $\beta$-bucket approximate equi-depth histogram.

The overview picture of the histogram processing framework is given in Figure~\ref{fig:jobs}. In the left of the figure, HDFS holds whole data including the new data, summary files, and created histograms according to user requests. The framework is in the right of the picture and Summarizer and Merger jobs take place in the framework. Every time, new data is pushed to HDFS, the Summarizer constructs its summary ($T$-bucket equi-depth histogram) and saves it to HDFS, again. When a user requests an equi-depth histogram of desired partitions (it can be any set of data partitions), the Merger processes the request by merging the related summaries of desired partitions and saves the merged final histogram to HDFS. These jobs can also be implemented in the Hive and Pig as user functions.

\begin{figure}[!t]
    \centering
    \includegraphics[width=0.4\textwidth]{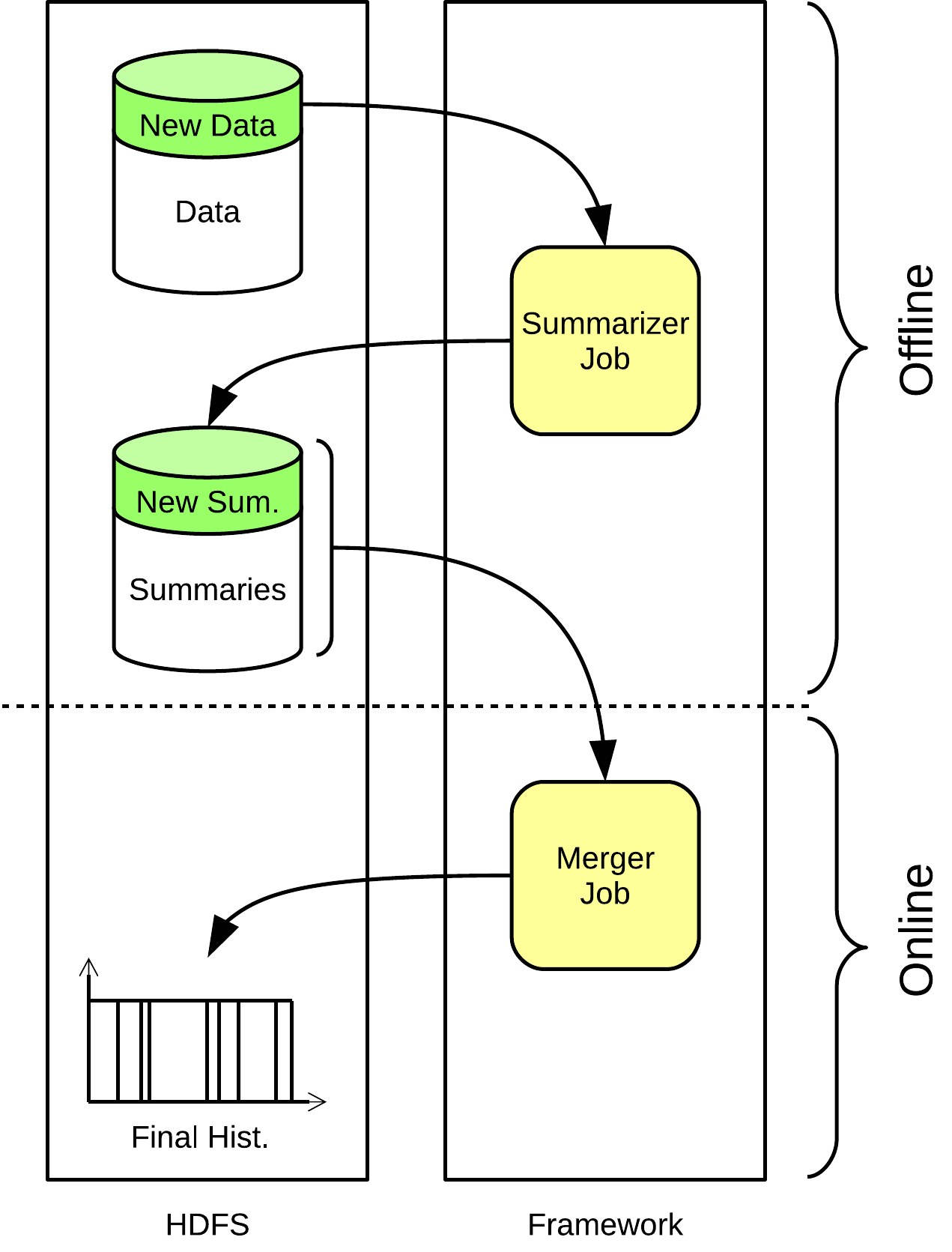}
    \caption{Overview of Proposed Framework}
  \label{fig:jobs}
\end{figure}

\section{Related Work}
\label{sec:rel}
Basically, histograms are used to get quick distribution of information from the given data. This quick information is used especially in database systems in computer science e.g. selectivity estimation to optimize queries, load balancing of join queries, and much more~\cite{ioannidis2003history}. There are different types of histograms and each type of histogram has different properties~\cite{poosala1996improved}. Exact histogram construction is not feasible when the data is too big or the data is frequently updated. In such cases, histograms are constructed from sampled data and/or maintained according to the updated data~\cite{gilbert2002fast}. This type of histograms are called approximate histograms rather than exact histograms. Approximate histogram construction from sampled data can be divided into two categories by sampling method~\cite{shi2013hedcplusplus} which are tuple-level sampling and block-level sampling. Tuple-level sampling method uses uniform-random-sampling to sample the data at tuple level to construct an approximate histogram at the desired error bound~\cite{gibbons1997fast,chaudhuri1998random}. Gibbons et al.~\cite{gibbons1997fast} proposed a sampling-based incremental maintenance approach of approximate histograms. The proposed approach, backing sample, keeps the sampled tuples up-to-date in a relation. A bound of the amount of the sampling size for a given error bound studied by Chaudhuri et al.~\cite{chaudhuri1998random} in addition to proposing an adaptive page sampling algorithm. The second method, block-level sampling, exemplifies the data according to an iterative cross-validation based approach~\cite{chaudhuri2001histogram,chaudhuri2004effective}. Chaudhuri et al.~\cite{chaudhuri2004effective} proposed a method for approximate histogram construction using an initial set of data and iteratively updated the constructed histogram until the histogram error is under the predetermined level. All of the proposed approaches above, however, are for single-node databases.

When the data is too big to handle in a single-node database, the data is distributed to multi-nodes. One of the well-known distributed data storage frameworks is Hadoop Distributed File System (HDFS)~\cite{borthakur2008hdfs} and the data processing framework of the stored data in the HDFS is Hadoop MapReduce~\cite{dean2010mapreduce}. The histogram construction of such distributed data is not well-studied and there is less work on histogram creation of distributed data than the ones on undistributed data. One of the adapted methods for contructing approximate histogram is tuple-level sampling. Okcan et al.~\cite{okcan2011processing} proposed a tuple-level sampling based algorithm to construct approximate equi-depth histograms for distributed data to improve processing theta-joins using MapReduce. The algorithm works as follows. In the map section of a MapReduce Job, a predefined number of tuples are selected randomly by scanning the whole data and outputted. The tuples are sorted and sent to the reducer. The reducer of the job determines and outputs the boundaries of equi-depth histograms. In~\cite{jestes2011building}, a method for approximate wavelet histogram construction for big data using MapReduce is proposed and an improved sampling method -ignoring low frequent sampled keys in splits- is given. The drawback of such histogram construction algorithms of distributed data using tuple-level sampling is that scanning the whole data is a time consuming process. Another approximate histogram construction method is proposed in~\cite{shi2013hedcplusplus}. This method also uses a sampling method named two-phase sampling which samples the whole data at block-level and constructs the approximate histogram and calculates the error. If the error is not in the desired error boundary, the additional sampling size needed is calculated and histogram construction process is repeated. The insufficiency of this method is that histogram is rebuilt for every new data and it requires a customized MapReduce framework. In this paper, we propose a novel approximate equi-depth histogram construction method with a log histogram monitoring framework that users can query the daily stored log files for their equi-depth histogram. In the proposed method, a MapReduce Job is scheduled to summarize the daily stored log files which means that the exact equi-depth histogram of each log file is constructed and stored in corresponding summary files and another MapReduce Job merges the summaries of intended log files for approximate equi-depth histogram construction.

\section{Experimental Results}
\label{sec:exp}
In this section, we will describe how we tested the proposed method. The method was implemented on Hadoop MapReduce framework and was tested on two different datasets. One of them is synthetic data with 155 million of tuples created by using Gumbel distribution for skewness to represent the response of the method for skewed data. The other one is 295 GB uncompressed real data which is taken from hourly page view statistics of Wikipedia. The data consists of approximately 5 billion tuples which belong to January 2015 and each tuple has 4 columns which are $language$, $page name$, $page views$ and $page size$. We used $page size$ for histogram construction. The proposed method ($merge$) is compared with corrected tuple level random sampling ($tuple$). By definition, bare tuple level random sampling collects tuples randomly and constructs histogram with collected tuples but doing so does not work well when the data is sparse at the edges. Therefore, we fix this problem by including the edge values to the collected tuples by default. As mentioned in Subsection~\ref{sec:prob}, histogram construction of the data coming from daily logs is an important issue and tuple level random sampling method is also unfavorable for constructing a histogram of a given time interval. Hereby, sampling stage of tuple level is run offline to compare time spendings.

All equi-depth histograms are build with bucket size of 254 as used by Oracle as default bucket size for histogram enhancement. We fundamentally run two types of tests to represent the effectiveness of the proposed method in terms of boundary and bucket size error and run time. The first test represents the results according to T changes which is daily exact histogram bucket size for the proposed method and sample size for tuple level sampling. The second test represents the run time efficiency of histogram construction for the changes in a given time interval.

Approximate histograms may have two types of error. One of them is that approximate histograms may not have the same bucket boundaries with the exact ones and the other is that bucket sizes of the approximate histograms may deviate from the exact ones. The former error is named as boundary error ($\mu_{b}$) and defined as follows:
\begin{equation}
\label{eqn:boundary_err}
\mu_{b} = \frac{B}{v_{max} - v_{min}} \sqrt{\frac{1}{B+1} \sum_{i=1}^{B+1} [b(i, H^{*}) - b(i, H)]^{2}}.
\end{equation}
where $B$ is the bucket size, $v_{max}$ and $v_{min}$ are maximum and minimum values in relation $R$ respectively, and the function $b(i, H)$ is the $i^th$ value of a given histogram $H$. $\mu_{b}$ is the standard deviation of boundary errors normalized with respect to the mean boundary length $(v_{max} - v_{min}) / B$. The latter error is named as size error ($\mu_{s}$) and formulated as follows:
\begin{equation}
\label{eqn:size_err}
\mu_{s} = \frac{B}{N} \sqrt{\frac{1}{B} \sum_{i=1}^{B} [s(i, H^{*}) - s(i, H)]^{2}}.
\end{equation}
where $N$ is the total number of elements in relation $R$ and function $s(i, H)$ is the size of the $i^th$ bucket of a given histogram $H$. $s(i, H)$ is equal to the mean bucket size $N/B$ for all $i$ values in the range of $1$ to $B$ if the given histogram $H$ is an equi-depth histogram. $\mu_{s}$ is the standard deviation of bucket size errors normalized with respect to the mean bucket size $N/B$.

\subsection{Effect of $T$}
\label{sub:effect_of_t}

\begin{figure}[!t]
\centering
\subfloat[Graph of $\mu_{b}$ against $T$ for real data]{\includegraphics[width=3in]{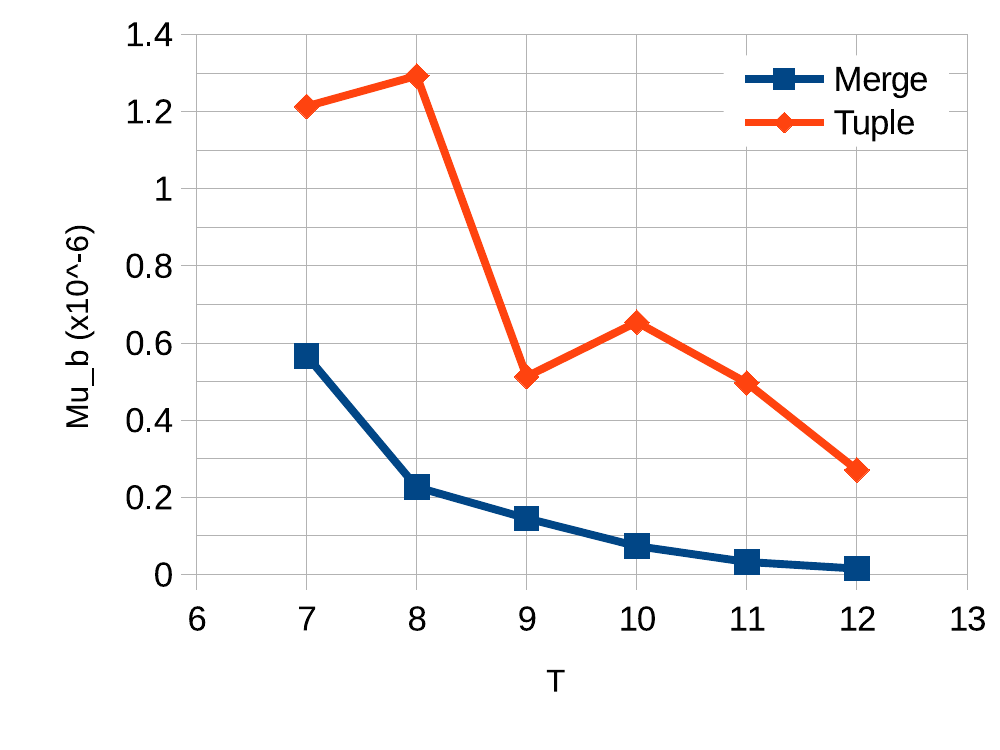}%
\label{fig:t_wiki_be}}
\hfil
\subfloat[Graph of $\mu_{s}$ against $T$ for real data]{\includegraphics[width=3in]{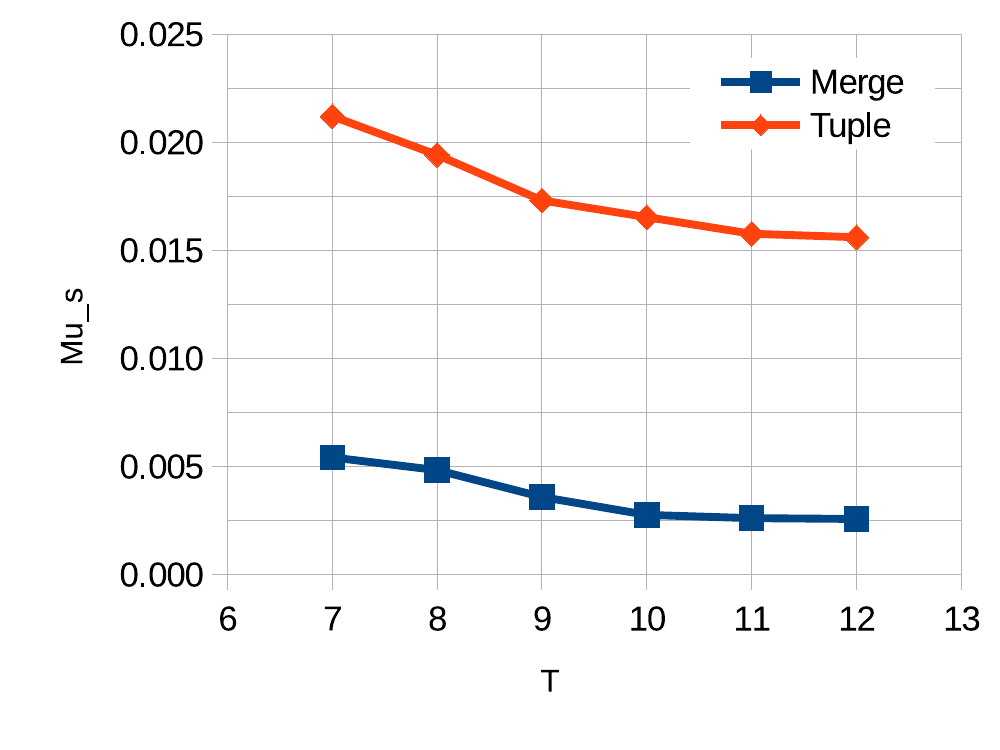}%
\label{fig:t_wiki_bse}}
\caption{Lin-log graphs of error metrics against $T$ ($B\times254\times2^{n}$) which is summary size in $merge$ method and sampling size in $tuple$ for real data}
\label{fig:t_wiki}
\end{figure}

\begin{figure}[!t]
\centering
\subfloat[Graph of $\mu_{b}$ against $T$ for skewed data]{\includegraphics[width=3in]{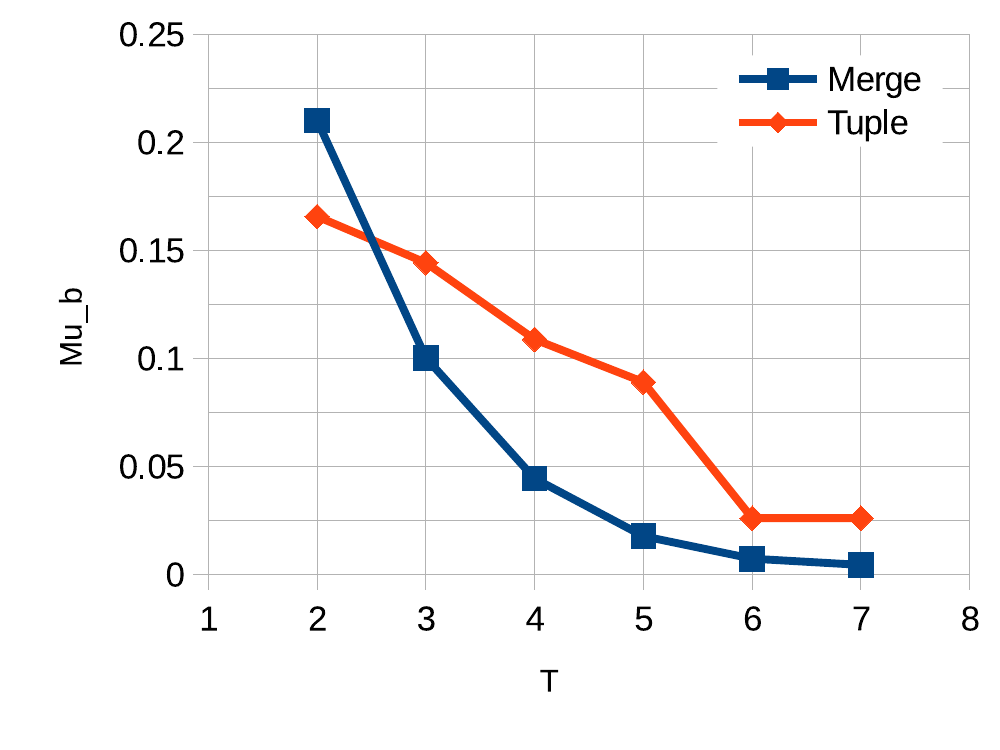}%
\label{fig:t_skewed_be}}
\hfil
\subfloat[Graph of $\mu_{s}$ against $T$ for skewed data]{\includegraphics[width=3in]{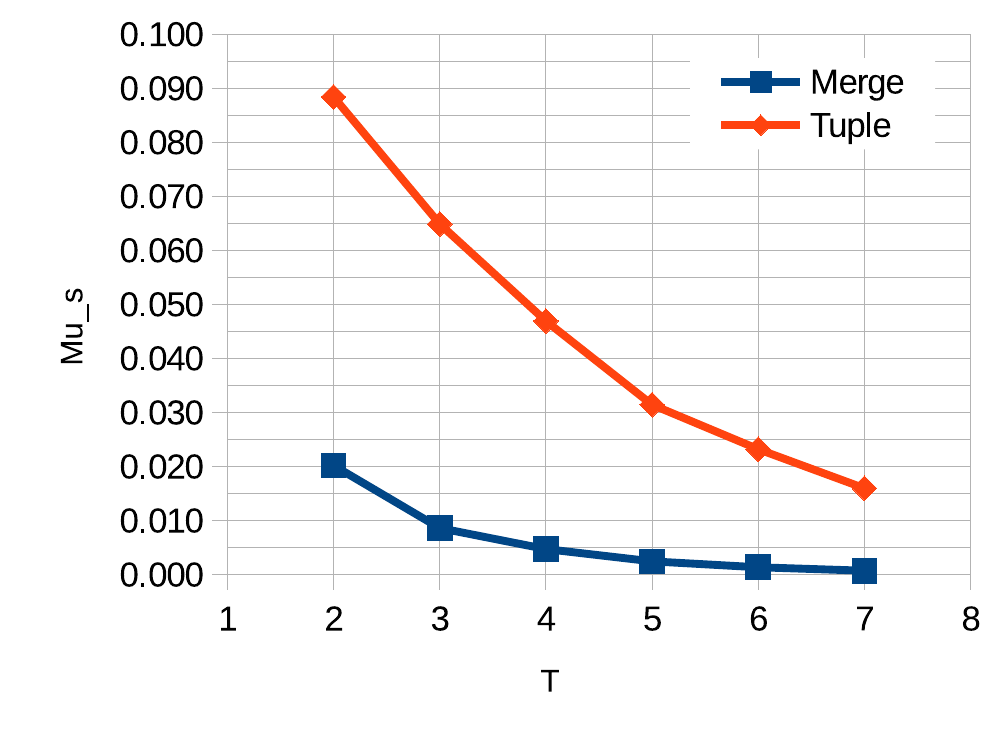}%
\label{fig:t_skewed_bse}}
\caption{Lin-log graphs of error metrics against $T$ ($B\times254\times2^{n}$) which is summary size in $merge$ method and sampling size in $tuple$ for skewed data}
\label{fig:t_skewed}
\end{figure}

Changing $T$ value effects the accuracy of the constructed approximate equi-depth histogram. The first experiment is run to show the effects of $T$ changes on both the proposed method and tuple level sampling. Figures~\ref{fig:t_wiki} and~\ref{fig:t_skewed} show the error graphs of the approximate equi-depth histograms constructed using the proposed method and tuple level sampling method. According to the graph in Figure~\ref{fig:t_wiki_be}, the constructed histogram by using $merge$ method for real data is at least 2 times more accurate in terms of boundary error $\mu_{b}$ than the one constructed using $tuple$ method. Moreover, the $\mu_{b}$ error for $tuple$ method is not consistent because of the randomness and the construction process must be repeated for consistency and this is not convenient because of the run time. For example, let us consider the graph of $\mu_{b}$ against T given in Figure~\ref{fig:t_wiki_be}. Notice that the $\mu_{b}$ error for $tuple$ method is not consistent. The expected result is that $\mu_{b}$ should decrease while $T$ increases. On the other hand, it is clearly seen from the graphs of $merge$ method in Figures~\ref{fig:t_wiki_be},~\ref{fig:t_wiki_bse},~\ref{fig:t_skewed_be}, and~\ref{fig:t_skewed_bse} that $\mu_{b}$ is a non-decreasing function of $T$. The reason of this consistency is the maximum error bound of $merge$ method described in Section~\ref{sec:algo} in detail. The mean running times for all methods are given in Table~\ref{tab:t_times}. Run times for merging daily summaries and samplings are nearly the same. But required time for the summarization stage of $merge$ method is more than the time for offline sampling stage of $tuple$ method. The reason of this time difference is that summarization is exact histogram construction and exact histogram construction for each data partition requires a complete MapReduce Job with a mapper and a reducer and the data comes from each mapper subjected to shuffle and sort phase. On the other hand, tuple level random sampling does not require a reducer because the randomly selected tuples would be stored directly without sorting. Because of this difference, summarization of each day for real data takes approximately 12 minutes while sampling takes 4 minutes. Although this time efficiency of $tuple$ method, all the daily summarizations and samplings are done offline, it makes $merge$ method convenient for real life applications.

\subsection{Effect of given time interval}
\label{sub:effect_of_days}

\begin{figure}[!t]
\centering
\subfloat[Graph of $\mu_{b}$ against merged \# of days for real data]{\includegraphics[width=3in]{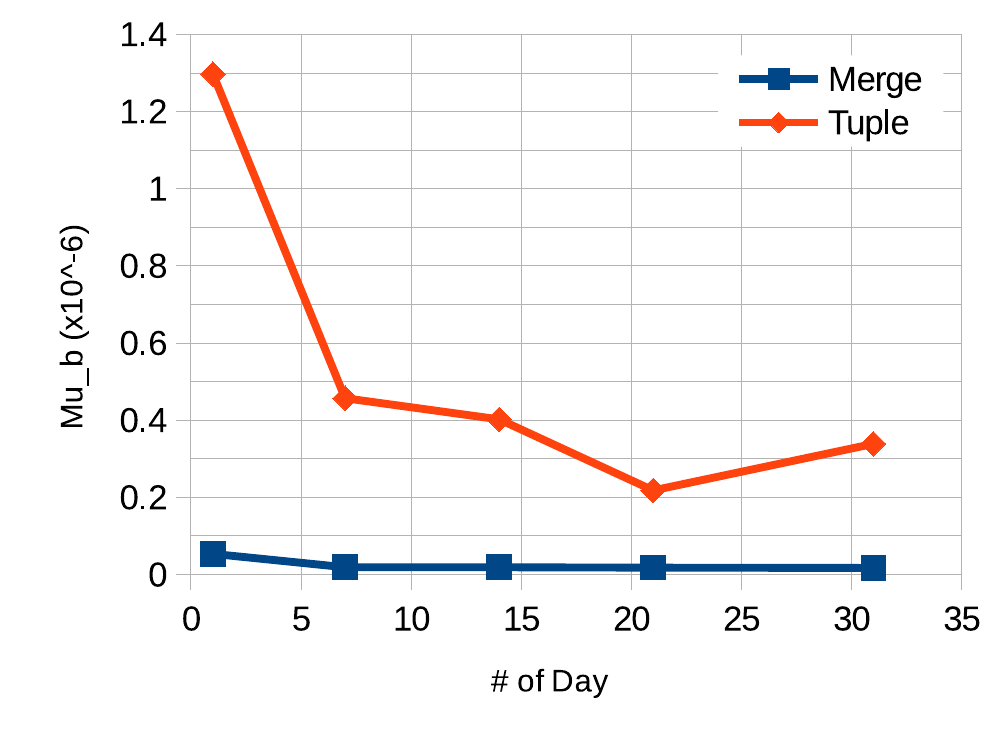}%
\label{fig:daily_wiki_be}}
\hfil
\subfloat[Graph of $\mu_{s}$ against merged \# of days for real data]{\includegraphics[width=3in]{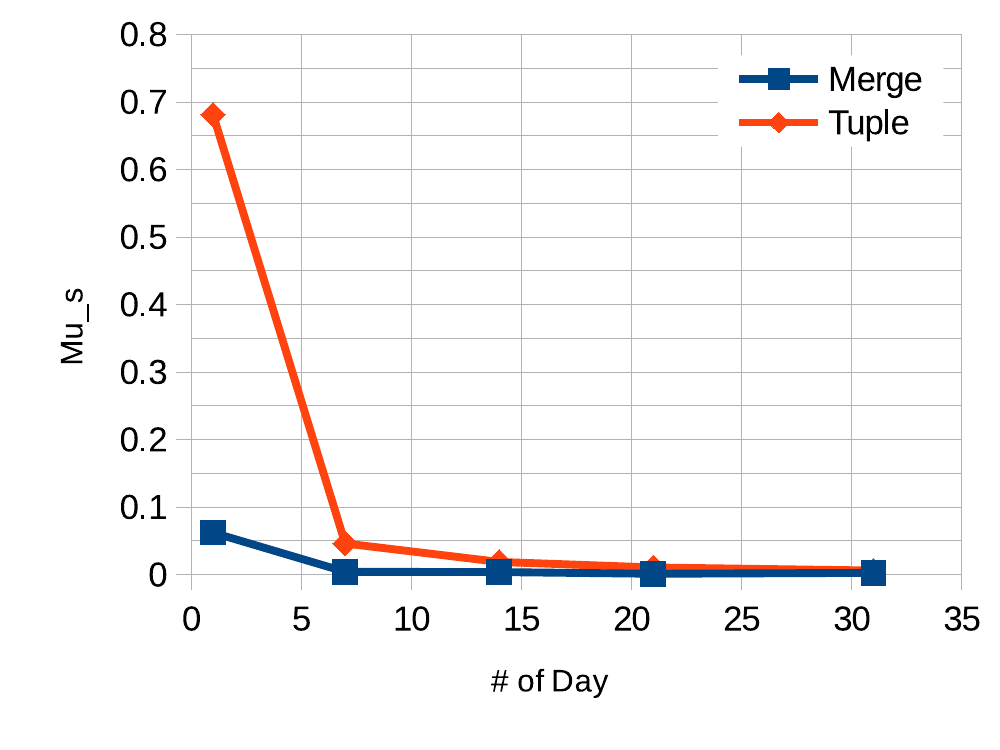}%
\label{fig:daily_wiki_bse}}
\caption{Graphs of error metrics against merged \# of days for real data}
\label{fig:daily_wiki}
\end{figure}

\begin{figure}[!t]
\centering
\subfloat[Graph of $\mu_{b}$ against merged \# of days for skewed data]{\includegraphics[width=3in]{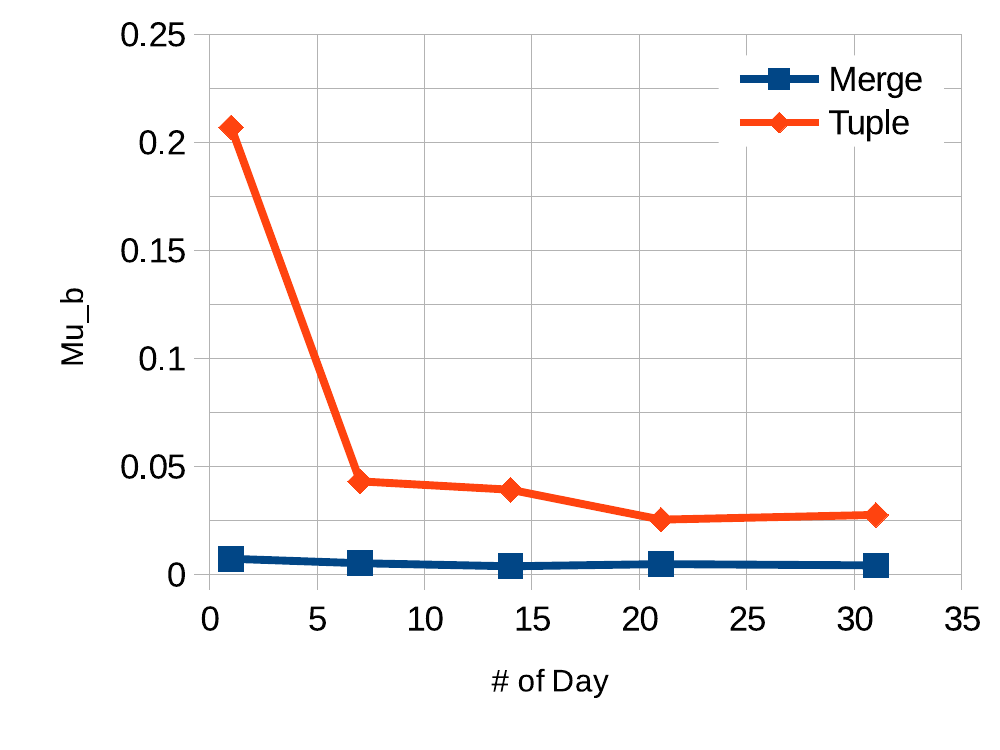}%
\label{fig:daily_skewed_be}}
\hfil
\subfloat[Graph of $\mu_{s}$ against merged \# of days for skewed data]{\includegraphics[width=3in]{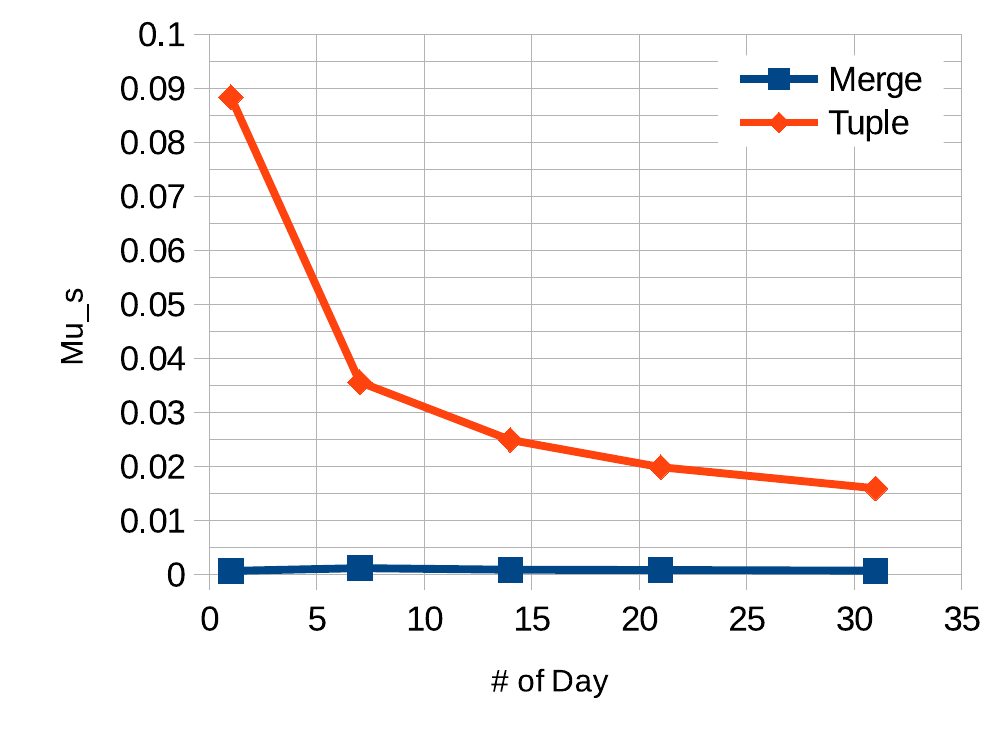}%
\label{fig:daily_skewed_bse}}
\caption{Graphs of error metrics against merged \# of days for skewed data}
\label{fig:daily_skewed}
\end{figure}
Histogram constructions according to the given time interval effect deviation of approximate histograms from the exact ones and running times. We compared $merge$ method with $tuple$ method with an experiment for 5 different time intervals (1 day, 1 week, 2 weeks, 3 weeks, and 1 month). $T$ value is taken to be $B\times254\times2^{12}$ for real data and to be $B\times254\times2^{7}$ for skewed data. All daily exact histogram and samplings are computed offline. In Figures~\ref{fig:daily_wiki} and~\ref{fig:daily_skewed}, graphs of error metrics against time intervals are given and it is clearly seen from the graphs that the proposed method produces more sensible histograms than the ones produced using tuple level random sampling. In particular, a real data histogram constructed using $merge$ method has at least 10 times less boundary error than the one constructed using $tuple$ method in Figure~\ref{fig:daily_wiki_be}. Besides, again the consistency is an issue for $tuple$ method as seen in Figures~\ref{fig:daily_wiki_be} and~\ref{fig:daily_skewed_be}. The graphs of running time against number of days are given in Figure~\ref{fig:daily_times}. More specifically, the compared methods ($merge$ and $tuple$) run in nearly the same time duration except offline parts.

\begin{figure}[!t]
\centering
\subfloat[Running time against merged \# of days for real data]{\includegraphics[width=3in]{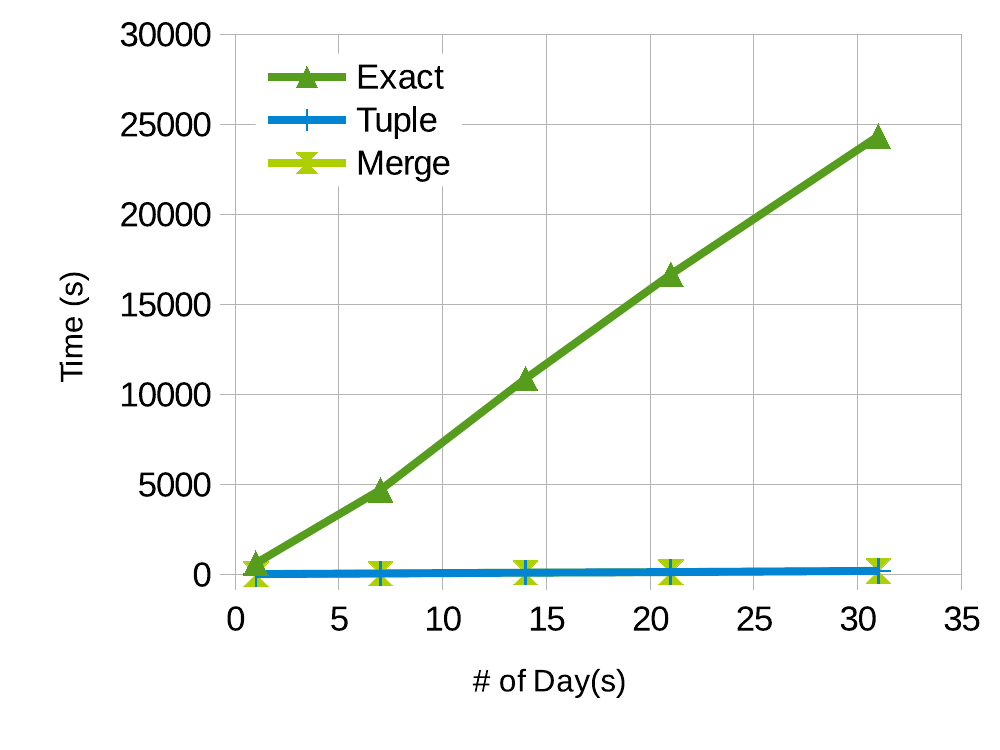}%
\label{fig:daily_wiki_time}}
\hfil
\subfloat[Running time against merged \# of days for skewed data]{\includegraphics[width=3in]{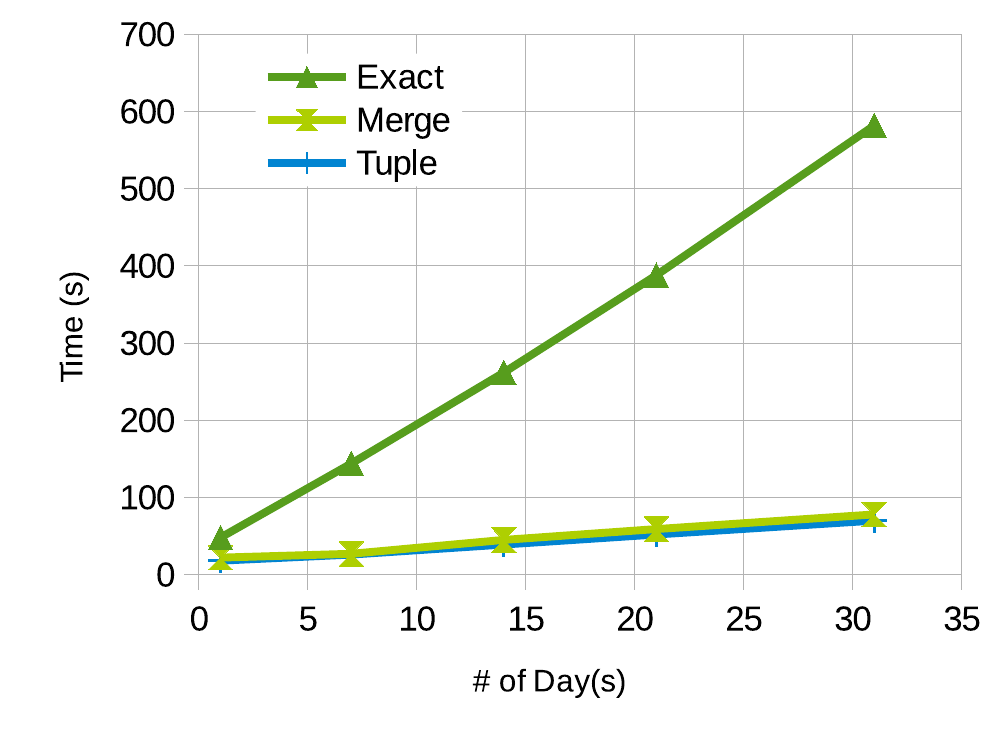}%
\label{fig:daily_skewed_time}}
\caption{Graphs of running time against merged \# of days}
\label{fig:daily_times}
\end{figure}

\begin{table}[!t]
\renewcommand{\arraystretch}{1.3}
\caption{Mean Running Times of Monthly Equi-depth Histogram Construction}
\label{tab:t_times}
\centering
\begin{tabular}{lrr}
\hline
Method                     & Real Data & Skewed Data \\ \hline
Exact Hist. Construction   & 24358     & 582         \\
Tuple with Online Sampling & 6169      & 68          \\
Summarizing for Each Day   & 725       & 18          \\
Merging of Daily Summaries & 117       & 73          \\
Sampling for Each Day      & 223       & 18          \\
Merging of Daily Samplings & 113       & 71          \\ \hline
\end{tabular}
\end{table}

\section{Conclusion}
\label{sec:conc}
In this paper, we proposed a novel approximate equi-depth histogram construction method by merging precomputed exact equi-depth histograms of data partitions. The method is implemented on Hadoop to demonstrate how it is applied to real life problems. The theoretical calculations and the experimental results showed that both the bucket size errors and total size error of any bucket range are bounded by a predefined error set by a user in terms of $T$ and $\beta$. In particular, the experimental results run on both real and synthetic data show that the constructed histograms using the proposed method ($merge$) are more accurate than the tuple level random sampling ($tuple$) with a cost of offline run time. In addition to the proposed merged based histogram construction method, we also proposed a novel histogram processing framework for the daily stored log files. 
This framework is crucial for fast histogram construction over a subset of a list of partitions on demand.
The time complexity and the incrementally updated nature of the proposed method makes it practical to be applied over real life problems.

\ifCLASSOPTIONcompsoc
  \section*{Acknowledgments}
\else
  \section*{Acknowledgment}
\fi

We would like to thank Enver Kayaaslan for his help for presentation.

\ifCLASSOPTIONcaptionsoff
  \newpage
\fi



%
\bibliographystyle{IEEEtran}
\bibliography{paper}

%

\begin{IEEEbiography}[{\includegraphics[width=1in,height=1.25in,clip,keepaspectratio]{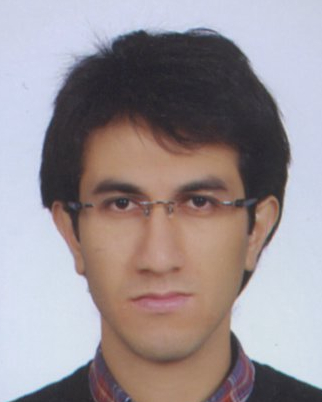}}]{Burak Y{\i}ld{\i}z}
received his BS degree in mechanical engineering and MS degree in computer engineering from TOBB University of Economics and Technology in 2013 and 2016, respectively. Large scale data management and data mining are in his major interest area.
\end{IEEEbiography}

\begin{IEEEbiography}[{\includegraphics[width=1in,height=1.25in,clip,keepaspectratio]{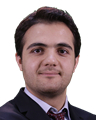}}]{Tolga B\"{u}y\"{u}ktan{\i}r}
got his BS degree in computer engineering from Erciyes University in 2014. He is now a master student at Turgut Ozal University in computer science. His research areas include data mining and IoT.
\end{IEEEbiography}

\begin{IEEEbiography}[{\includegraphics[width=1in,height=1.25in,clip,keepaspectratio]{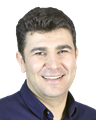}}]{Dr. Fatih Emekci}
got his MS and PhD from the Department of Computer Science at the University of California, Santa Barbara in 2006. He was at Oracle working with query optimizer team and at Linkedin working on data scalability problems until 2012. He has been an associated professor of computer science at Turgut Ozal University, Ankara, Turkey since 2013.
\end{IEEEbiography}



\end{document}